\documentclass[
final
, nomarks
]{dmtcs-episciences}

\usepackage[utf8]{inputenc}
\usepackage{subfigure}
\usepackage[noend, boxed]{algorithm2e}
\usepackage{amsmath,amssymb,amsfonts}
\usepackage{amsthm}
\usepackage{multirow}
\usepackage{mathbbol}
\usepackage{comment}

\newtheorem{thm}{Theorem}
\newtheorem{theorem}[thm]{Theorem}
\newtheorem{lemma}[thm]{Lemma}

\usepackage[round]{natbib}

\author[Mojtaba Ostovari and Alireza Zarei]{Mojtaba Ostovari\affiliationmark{1} \and Alireza Zarei\affiliationmark{1}
}
\title[Improved Approximations for Correlation Clustering]
{Improved Combinatorial Approximations for Weighted Correlation Clustering}
\affiliation{Department of Mathematical Sciences, Sharif University of Technology, Tehran, Iran}
\begin{document}
\publicationdata{vol. 27:2}{2025}{16}{10.46298/dmtcs.12424}{2023-10-17; 2023-10-17; 2024-07-15; 2025-04-14}{2025-06-03}
\maketitle
\begin{abstract}
\vspace{0.18cm}
We present combinatorial approximation algorithms for the weighted correlation clustering problem. In this problem, we have a set of vertices and two weight values for each pair of vertices, denoting their difference and similarity. The goal is to cluster the vertices with minimum total intra-cluster difference weights plus inter-cluster similarity weights.
We present two results for weighted instances with $n$ vertices:
\begin{itemize}
    \item A randomized 3-approximation combinatorial algorithm for instances that satisfy probability constraints, running in $O(n^2)$ time. 
    This improves the $O(n^6)$ running time of the previous best-known combinatorial approximation, a 3-approximation algorithm, introduced by Chawla et al. (2015).
    \item A randomized 1.6-approximation combinatorial algorithm for instances that satisfy probability and triangle inequality constraints, running in $O(n^2)$ time. 
    This improves the longstanding combinatorial 2-approximation of Ailon et al. (2008) while matching its running time.
\end{itemize}
\end{abstract}
\section{Introduction}
In the correlation clustering problem, we have a set of vertices, and each pair of vertices are either similar or different. The goal is to find a clustering for the vertices with an arbitrary number of clusters that minimizes the number of pairs of similar vertices in different clusters or different pairs inside the same cluster. In this paper, we consider the weighted version, called weighted correlation clustering problem. In this version, there is a set of vertices and for each pair of vertices, there are two weighted edges that indicate the weight of their similarity and difference. We want to partition these vertices into a set of clusters in such a way that minimizes the sum of the difference weights over all pairs that are in the same cluster plus, the sum of the similarity weights over all pairs that are not in the same cluster. The unweighted version is a special case of the weighted version in which the weights are either zero or one.

Let $V$ be the set of vertices, and $D$ be the set of \textit{unordered}\footnote{By unordered we meant that there is no distinction between $(i,j)$ and $(j,i)$ pairs.} pairs of $V$.
Formally, an instance of the weighted correlation clustering can be represented by a weighted graph $(V, w^{+}, w^{-})$
where for any $(i, j) \in D$, there are two edges with weights $w^{+}_{ij}\geq0$ and $w^{-}_{ij}\geq0$ which indicate similarity and difference weights for this pair, respectively. Note that as $(i,j)$ is \textit{unordered} $w^{+}_{ij}=w^{+}_{ji}$ and $w^{-}_{ij}=w^{-}_{ji}$.
Define $c_{ij}$ as the contributed cost of a pair $(i,j)\in D$ for clustering $C$ that is
$w^{-}_{ij}$ if the vertices are in the same cluster and 
$w^{+}_{ij}$ if they belong to different clusters.
The goal is to find a clustering $C$ with an arbitrary number of clusters that minimizes the total cost defined as 
$$
cost(C) = \sum_{(i, j) \in D} c_{ij}.
$$

An instance of this problem satisfies the probability constraints if 
$w^{+}_{ij} + w^{-}_{ij} = 1$ for all $(i, j) \in D$.
In this case, there is a fractional similarity/difference between any two vertices. 
An instance satisfies the triangle inequality If $w^{-}_{ik} \leq w^{-}_{ij} + w^{-}_{jk}$ for any distinct $i, j, k \in V$.
The instance is unweighted if it satisfies probability constraints and $w^{+}_{ij}, w^{-}_{ij} \in \{0, 1\}$. This unweighted version is also known as the Cluster Editing problem~(\cite{Shamir}). 

The weighted correlation clustering problem with probability and triangle inequality constraints has an application in the Clustering Aggregation problem~(\cite{Clustering_Aggregation}). 
In this problem, we have a set of clusterings $C_1, C_2, \cdots, C_n$ on a given set of vertices $V$ and the objective is to find a clustering $C$ that minimizes $\sum_{i=1}^{n} d(C, C_i)$ where $d(C, C_i)$ is the number of distinct pairs of vertices that exist in the same cluster of $C$, but in different clusters of $C_i$, or vice versa. This problem can be reduced to the weighted correlation clustering problem with probability and triangle inequality constraints: consider the instance $(V, w^+, w^-)$ where $w^{+}_{ij}$ equals the number of clusterings in which vertices $i$ and $j$ belong to the same cluster divided by $n$, and $w^{-}_{ij}$ equals the number of clusterings in which vertices $i$ and $j$ are in different clusters divided by $n$. Clearly, this instance satisfies both probability and triangle inequality constraints. Furthermore, an optimal solution for this problem is also an optimal solution for the instance $C_1, C_2, \cdots, C_n$ of the Clustering Aggregation problem.

\subsection{Related works}
Correlation clustering was first introduced by \cite{Bansal04}. They showed that the decision version of this problem is NP-complete. Furthermore, they designed a constant factor approximation algorithm and a PTAS for the maximization version of this problem. In the maximization problem, the goal is to find a clustering that maximizes the number of pairs of similar vertices inside the same cluster and different pairs in different clusters. However, the story is different for the minimization version. In the remainder of this paper, any reference to this problem specifically refers to the minimization version. 

\cite{DEMAINE} showed that the weighted correlation clustering when the instance's graph is not complete is APX-hard (it does not have a PTAS unless P=NP) and gave a $O(\log n)$-approximation algorithm for this problem.

\cite{CHARIKAR05} showed that the weighted correlation clustering problem is APX-hard even on complete graphs. They also gave a 4-approximation algorithm for this problem using a pivoting method. Briefly, this method continues as a sequence of steps starting from step $i=0$ with $V_0=V$. Then, in the $i^{th}$ step, a pivot vertex $p_i$ is chosen uniformly at random from the set of vertices $V_i$; a new cluster $C_i=\{p_i\}$ is initialized; some of the other vertices of $V_i$ are assigned to $C_i$ based on some criteria on these vertices and $p_i$; and finally, $V_{i+1}$ of the next step is defined to be $V_{i+1}=V_i-C_i$. This algorithm terminates at step $t$ when $V_{t+1} = \emptyset$, and the output clusters are $C_0,\dots,C_t$.
The criteria used to add vertices to clusters are based on the values of variables of the optimal solution of the corresponding relaxed linear programming formulation shown in LP \ref{LP:CWCC}. In this formulation we have a set of variables $\{ x_{ij} \in \{0,1\} \,|\,(i,j) \in D \}$, if vertices $i$ and $j$ are in the same cluster then $x_{ij} = 0$, otherwise $x_{ij} = 1$. For any distinct $i, j, k \in V$ if $x_{ij}=0$ and $x_{jk}=0$, then vertices $i$, $j$ and $k$ are in the same cluster, so~$x_{ki}$ must be 0. This condition has been expressed by the triangle inequality that is enforced by the first set of constraints of this LP. Recall that the members in $D$ are \textit{unordered} meaning that $x_{ij}=x_{ji}$. The constraints $x_{ij} \in \{0,1\}$ cannot be expressed as linear programming constraints, instead, the second set of constraints of LP~\ref{LP:CWCC} is used to obtain a linear programming relaxation. The best-known lower bound for the integrality gap of this LP is 2. However, \cite{CHARIKAR05} obtained their 4-approximation algorithm using the solution of this LP in their pivoting method.
\begin{align}
\boxed{
	\label{LP:CWCC}
	\begin{array}{ll@{}ll}
		\textrm{minimize}  
		& \displaystyle\sum\limits_{i<j}  (1-x_{ij})w^{-}_{ij} + x_{ij}w^{+}_{ij} \\
		\textrm{subject to}
		& x_{ik} \leq x_{ij} + x_{jk}  \hspace{1cm} 
		&\textrm{
			for all distinct $i, j, k \in V$
		} \\
		& 0 \leq x_{ij} \leq 1              
		&\textrm{
			for all $(i, j) \in D$
		}
	\end{array}
}
\end{align}

Using a better pivoting method, \cite{Ailon08} introduced another LP-based algorithm whose approximation factor is 2.5 when the instance satisfies the probability constraints, and 2 if it satisfies both the probability and the triangle inequality constraints. The criteria used in their algorithm to assign the vertices to clusters are also based on the values of variables of LP~\ref{LP:CWCC}, however in a better way which leads to this improvement.
They also gave a 5-approximation combinatorial (without requiring solving the LP) algorithm using another pivoting method when the instance satisfies the probability constraints with $O(n^2)$ running time. 
For the unweighted instances it is a 3-approximation algorithm, and for the instances that satisfy the triangle inequality constraints in addition to probability constraints, it is a 2-approximation algorithm. 
To the best of our knowledge, this is the best randomized combinatorial approximation algorithm for these instances.
In this algorithm, in the $i^{th}$ step of the pivoting method, a vertex $v\in V_i$ is added to $C_i$ if $w^{+}_{v p_i} \geq w^{-}_{v p_i}$.
For the unweighted instances (cluster editing problem), \cite{veldt} recently gave a deterministic combinatorial 6-approximation algorithm.

Chawla \textit{et al.} by using an improved pivoting method and solving the LP instance, gave a 2.06-approximation algorithm when the instance satisfies the probability constraints, which almost matches the integrality gap~(\cite{NearOptimalClustering}). Their algorithm has 1.5 approximation factor when the instance satisfies the triangle inequality constraints in addition to the probability constraints.

\cite{Cohen-Addad} proposed a (1.994+$\epsilon$)-approximation for this problem which is based on different LP formulations and the $O(1/\epsilon^2)$ rounds of the Sherali-Adams hierarchy. To round a solution to the Sherali-Adams relaxation, they adapt the correlated rounding developed for CSPs~(\cite{BRS,GS}). Very soon \cite{Cohen-Addad2} improved it to 1.73 by using a new linear programming formulation and relying on the Sherali-Adams hierarchy too. These algorithms require a solution of an LP with $O(n^{1/\epsilon^{\Theta(1)}})$ variables.

Chawla \textit{et al.} also presented a reduction which turns any $\alpha-$approximation algorithm for the unweighted instances to an $\alpha-$approximation algorithm for instances that satisfy the probability constraints. 
In this reduction, they first pick a sufficiently large integer number $N = O(n^2)$. Then, any vertex $v$ of the instance is replaced by $N$ new vertices $v_1, v_2,\cdots, v_N$. 
For each pair of vertices $u_i$ and $v_j$ such that $u\neq v$, with probability $w^{+}_{uv}$,
we have $w^{+}_{u_i v_j}=1$ and $w^{-}_{u_i v_j}=0$, otherwise (with probability $w^{-}_{uv}$) 
we have $w^{+}_{u_i v_j}=0$ and $w^{-}_{u_i v_j}=1$.
Also, for each pair of vertices $v_i$ and $v_j$, 
we have $w^{+}_{v_i v_j} = 1$ and~$w^{-}_{v_i v_j} = 0$.
The result is a new unweighted instance with $O(n^3)$ vertices. Consider an $\alpha-$approximation solution for this instance. 
For any vertex $v$, one of the vertices $v_1, v_2,\cdots, v_N$ is turned to $v$ uniformly at random, and the others are removed. This new clustering is an $\alpha-$approximation solution for the first instance.
Therefore, using this reduction and the 3-approximation combinatorial algorithm with $O(n^2)$ running time for the unweighted instances introduced by \cite{Ailon08}, we have a 3-approximation combinatorial algorithm for the weighted instances that satisfy the probability constraints. To the best of our knowledge, it is the best combinatorial approximation algorithm for this problem but its running time is $O(n^6)$, because in this reduction, we need a solution of an unweighted instance with $O(n^3)$ vertices.

\subsection{Our Contributions}
In this article, we introduce a combinatorial randomized approximation algorithm with $O(n^2)$ run time for the mentioned problem:
a 3-approximation algorithm for the weighted correlation clustering problem that satisfies probability constraints, and 
a 1.6-approximation algorithm when the triangle inequality constraints are satisfied as well.

Our algorithm uses the pivoting method introduced by \cite{CHARIKAR05} to approximate this problem. Later, using this method \cite{Ailon08} and \cite{NearOptimalClustering} gave better approximations for this problem later. The analysis we will give for our algorithm is called triangle-based analysis, which was introduced by \cite{Ailon08} and later, has been used by \cite{NearOptimalClustering} to analyze their algorithms.
 As the main contribution of our work, we define a function on the edge weights and use this function to assign the vertices to the clusters. Then, we do not need to solve an LP and this enables our method to run much faster than LP-based algorithms. Precisely, the running time of our algorithm is $O(n^2)$ for a set of $n$ vertices while the LP-based method for $\alpha=2.06$ must solve an LP with $\Theta(n^2)$ variables, which makes them impractical for $n>1000$. Although the approximation factors of our algorithm are slightly greater than the best-known LP-based algorithms, its smaller running time makes it superior and practical for real applications. 
Table \ref{tabel:summarize} compares our algorithms to the previous best-known combinatorial algorithms for this problem.
\begin{table}[h]
\caption{Comparison of our combinatorial algorithms with the previous best-known combinatorial algorithms.}
\begin{tabular}{|l|cc|cc|}
\hline
\multicolumn{1}{|c|}{\multirow{2}{*}{\textbf{Problem}}}                                                                              & \multicolumn{2}{c|}{\textbf{\begin{tabular}[c]{@{}c@{}}Previous Best-Known \\ Combinatorial Algorithms\end{tabular}}}                             & \multicolumn{2}{c|}{\textbf{\begin{tabular}[c]{@{}c@{}}Our Combinatorial \\ Algorithms\end{tabular}}} \\ \cline{2-5} 
\multicolumn{1}{|c|}{}                                                                                                               & \multicolumn{1}{c|}{\begin{tabular}[c]{@{}c@{}}Approximation\\ factor\end{tabular}} & Run time                                                    & \multicolumn{1}{c|}{\begin{tabular}[c]{@{}c@{}}Approximation\\ factor\end{tabular}}     & Run time    \\ \hline
\begin{tabular}[c]{@{}l@{}}Weighted Correlation Clustering\\ with Probability Constraints\end{tabular}                           & \multicolumn{1}{c|}{\begin{tabular}[c]{@{}c@{}}5 (\cite{Ailon08})\\ 3 (\cite{NearOptimalClustering})\end{tabular}}                  & \begin{tabular}[c]{@{}c@{}}$O(n^2)$\\ $O(n^6)$\end{tabular} & \multicolumn{1}{c|}{3}                                                                  & $O(n^2)$    \\ \hline
\begin{tabular}[c]{@{}l@{}}Weighted Correlation Clustering\\ with Probability and\\ Triangle Inequality Constraints\end{tabular} & \multicolumn{1}{c|}{2 (\cite{Ailon08})}                                                              & $O(n^2)$                                                    & \multicolumn{1}{c|}{1.6}                                                                & $O(n^2)$    \\ \hline
\end{tabular}
\label{tabel:summarize}
\end{table}
\section{Our Clustering Algorithm}\label{sub_section:the_algorithm}
Our algorithm for weighted correlation clustering needs a rounding function 
$f: V \times V \rightarrow [0,1]$ such that $f(a, b)=f(b, a)$ for any $a,b\in V$, to decide in the $i^{th}$ step of the algorithm, whether a vertex should be added to the cluster $C_i$ or not. 
Precisely, the algorithm starts with $V_0 = V$ and in each step $i\geq0$, a vertex $p \in V_i$ is selected uniformly at random as a pivot and a new cluster $C_i=\{p\}$ is initialized. Then, each vertex $u \in V_i-\{p\}$ is added to $C_i$ with probability $1-f(u,p)$. The algorithm terminates when $V_i=\emptyset$. Finally, the output clusters will be $C_0, \dots, C_{i-1}$.
For simplicity, we use $f^{-}_{up}$ and $f^{+}_{up}$ instead of $f(u, p)$ and $1-f(u, p)$.
The pseudocode of our approximation algorithm for weighted correlation clustering, \textsc{QuickCluster} is depicted in Algorithm \ref{alg:QuickCluster}.

If the input instance satisfies the probability constraints, we define $f^{-}_{ij}$ to be $w^{-}_{ij}$, and
if it satisfies the triangle inequality constraints in addition to the probability constraints, we set $f^{-}_{ij}$ as $h(w^{-}_{ij})$, where
$$ 
h(w) = 
\begin{cases}
0 & \text{
	if 
	$\; w \in I_1 = [0, 0.35]$
} \\
\dfrac{25}{7}w - \dfrac{5}{4} & \text{
	if 
	$\; w \in I_2 = [0.35, 0.63]$
} \\
1 & \text{
	if 
	$\; w \in I_3 = [0.63, 1]$
}
\end{cases}
$$
\begin{algorithm}
\newcommand{\myfuncsty}[1]{\textsc{\textbf{#1}}}
\SetKwData{VL}{$V_L$}\SetKwData{VR}{$V_R$}
\SetFuncSty{textbf}
\SetKwFunction{Algorithm}{\textsc{QuickCluster}}
\SetKwProg{Fn}{}{:}{}
\Fn{\Algorithm{$(V, w^{+}, w^{-})$}}{
    set $i=0$ and $V_0=V$\;
    \While{$V_i \neq \emptyset$}{
	    pick $p \in V_i$ uniformly at random as pivot\;
	    make cluster $C_i=\{p\}$\;
		\For{each $u\in V_{i}-\{p\}$}{\label{forins}
			with probability $f^{+}_{up}$ add $u$ to $C_i$\;
		}
		$V_{i+1}=V_{i}-C_{i}$\;
		$i=i+1$\;
    }
\Return $C_0, \dots, C_{i-1}$
}
\caption{Pseudocode of \textsc{QuickCluster}.}
\label{alg:QuickCluster}
\end{algorithm}

In the rest of this section, we analyze the cost of our algorithm and compute its approximation factor. To do this, although we do not solve an LP, we compare the cost of our approximation solution to the optimal solution of the corresponding LP as a lower bound for the optimal cost.

Remind that $c_{ij}$ which is the contributed cost of a pair of vertices $i$ and $j$ in the total cost of a solution is either $w^{+}_{ij}$ or $w^{-}_{ij}$ and depends on whether these vertices belong to different clusters or the same one. We say the contributed cost is \textit{controlled} if in the latest step $l$ of the while loop in \textsc{QuickCluster} where both $i$ and $j$ exist in $V_l$, either $i$ or $j$ is picked as a pivot. Otherwise, the cost is called \textit{uncontrolled} which happens when another vertex $k$ in a step $l$ is chosen as the pivot and at least one of the vertices $i$ or $j$ is added to~$C_l$.

Assume that the output clusters of our algorithm are $C_0, \dots, C_t$. Let $A_{ij[k]}$ be the event in which for a step $r \in \{0,\dots,t\}$ in our algorithm, vertices $i$, $j$, and $k$ exist in $V_{r}$ and $k$ has been selected as this step pivot. Since the pivot is chosen uniformly at random, events $A_{ij[k]}$, $A_{ki[j]}$, and $A_{jk[i]}$ happen with equal probabilities when these three vertices exist in $V_{r}$.
Let 
$A_{ij[k]}^{-} \subset A_{ij[k]}$ be the event that $i$ and $j$ are added to $C_r$. 
Therefore when this event happens 
$\Pr[A_{ij[k]}^{-}] = \Pr[A_{ij[k]}] f^{+}_{ik}f^{+}_{jk}$ and 
$c_{ij}=w^{-}_{ij}$. 
Analogously, 
$A_{ij[k]}^{+} \subset A_{ij[k]}$ is the event that only one of $i$ and $j$ is added to $C_r$. 
Therefore when this event happens 
$\Pr[A_{ij[k]}^{+}] = \Pr[A_{ij[k]}]
(f^{+}_{ik}f^{-}_{jk} + f^{-}_{ik}f^{+}_{jk})$ and 
$c_{ij}=w^{+}_{ij}$.
Note that in both these cases, the contributed cost of the pair of vertices $i$ and $j$ is \textit{uncontrolled}. Therefore, the expected value of the \textit{uncontrolled} cost of the vertices $i$ and $j$ for the specific vertex $k$ is 
$$
\Pr[A_{ij[k]}^{-}]w^{-}_{ij} + Pr[A_{ij[k]}^{+}]w^{+}_{ij}=
\Pr[A_{ij[k]}]f^{+}_{ik}f^{+}_{jk}w^{-}_{ij} + 
\Pr[A_{ij[k]}](f^{+}_{ik}f^{-}_{jk} + f^{-}_{ik}f^{+}_{jk}) w^{+}_{ij}.
$$

Define $A_{ij[k]}^{*} := A_{ij[k]}^{-} \cup A_{ij[k]}^{+}$. The cost of the pair $i$ and $j$ is \textit{uncontrolled} if and only if the event $\bigcup_{k \in V-\{i, j\}} A_{ij[k]}^{*}$ happens.
Events $A_{ij[k]}^{*}$ for all $k \in V-\{i, j\}$ are disjoint because if one of these events happens, it sets $i$, $j$ or both of them to a cluster and the other events can no longer happen. Therefore 
$$\Pr[\bigcup_{k \in V-\{i, j\}} A_{ij[k]}^{*}] = 
\sum_{k \in V-\{i, j\}} \Pr[A_{ij[k]}^{*}].$$
Let $QC$ be the clustering returned by \textsc{QuickCluster}.
We define random variables $C_{QC}$ and $UC_{QC}$ as the sum of \textit{controlled} and \textit{uncontrolled} costs in $QC$, respectively. Then 
\begin{align*}
    E[ cost(QC) ] = E[C_{QC} ] + E[UC_{QC} ].
\end{align*}
The expected value of $UC_{QC}$ and $C_{QC}$ are computed as follows.
\begin{align}\label{eq:b_qcr}
E[UC_{QC}] =& 
\sum_{i<j} \sum_{k \in V-\{i,j\}}
\Pr[A_{ij[k]}^{-}] w^{-}_{ij} + \Pr[A_{ij[k]}^{+}] w^{+}_{ij}
\nonumber\\ =&
\sum_{i<j} \sum_{k \in V-\{i,j\}}
\Pr[A_{ij[k]}]
\big(
f^{+}_{ik}f^{+}_{jk} w^{-}_{ij} +
(f^{+}_{ik}f^{-}_{jk} + f^{-}_{ik}f^{+}_{jk}) w^{+}_{ij}
\big)
\nonumber\\ =&
\sum_{i<j<k} 
\Pr[A_{ij[k]}] \phi_{ijk}
\end{align}
where
\begin{align*}
\phi_{ijk} =&
f^{+}_{ik}f^{+}_{jk} w^{-}_{ij} +
(f^{+}_{ik}f^{-}_{jk} + f^{-}_{ik}f^{+}_{jk}) w^{+}_{ij} \\&+
f^{+}_{ji}f^{+}_{ki} w^{-}_{jk} +
(f^{+}_{ji}f^{-}_{ki} + f^{-}_{ji}f^{+}_{ki}) w^{+}_{jk} \\&+
f^{+}_{kj}f^{+}_{ij} w^{-}_{ki} +
(f^{+}_{kj}f^{-}_{ij} + f^{-}_{kj}f^{+}_{ij}) w^{+}_{ki}
\end{align*}
and
\begin{align}\label{eq:g_qcr}
E[C_{QC}] =& 
\sum_{i<j} 
\big( 1 - \sum_{k \in V-\{i,j\}} 
\Pr[A_{ij[k]}^{*}] \big)
\big( f^{+}_{ij} w^{-}_{ij} + f^{-}_{ij} w^{+}_{ij} \big).
\end{align}

To obtain an upper bound for the approximation factor of the cost of the approximation solution $QC$, we use the optimal solution of LP \ref{LP:CWCC} as a lower bound for the cost of the optimal solution and compare the cost of $QC$ to this value. Let $OPT$ be the optimal clustering for the weighted correlation clustering problem. Trivially $cost(LP)\leq cost(OPT)$ and if $cost(QC)\leq \alpha \cdot cost(LP)$ then $QC$ is an $\alpha$-approximation solution.
Let $cost(LP)$ be the optimal value of LP \ref{LP:CWCC} that is 
$\sum_{i<j} (1-x_{ij})w^{-}_{ij} + x_{ij}w^{+}_{ij} $. We can rewrite this cost as
\begin{align*}
   cost(LP) & =  \sum_{i<j} (1-x_{ij})w^{-}_{ij} + x_{ij}w^{+}_{ij} \\&
   =\sum_{i<j} \Big(\sum_{k \in V-\{i,j\}} \Pr[A_{ij[k]}^{*}] + 1 - \sum_{k \in V-\{i,j\}} \Pr[A_{ij[k]}^{*}] \Big)\big((1-x_{ij})w^{-}_{ij}+x_{ij}w^{+}_{ij}\big) \\&
    =\sum_{i<j} \sum_{k \in V-\{i,j\}} \Pr[A_{ij[k]}^{*}] \big((1-x_{ij})w^{-}_{ij}+x_{ij}w^{+}_{ij}\big) \\&\quad
    + \sum_{i<j} \Big( 1- \sum_{k \in V-\{i,j\}} \Pr[A_{ij[k]}^{*}] \Big)
\big((1-x_{ij})w^{-}_{ij}+x_{ij}w^{+}_{ij}\big)\\&
    = UC_{LP} + C_{LP}
\end{align*}
where
\begin{align}\label{eq:b_lp}
UC_{LP} =& 
\sum_{i<j} \sum_{k \in V-\{i,j\}} 
\Pr[A_{ij[k]}^{*}] \big((1-x_{ij})w^{-}_{ij}+x_{ij}w^{+}_{ij}\big) = 
\sum_{i<j<k} \Pr[A_{ij[k]}] \psi_{ijk}
\end{align}
for
\begin{align*}
\psi_{ijk} =& 
(f^{+}_{ik}f^{+}_{jk} + f^{+}_{ik}f^{-}_{jk} + f^{-}_{ik}f^{+}_{jk}) 
\big(x_{ij}w^{+}_{ij}+(1-x_{ij})w^{-}_{ij}\big) \\&+
(f^{+}_{ji}f^{+}_{ki} + f^{+}_{ji}f^{-}_{ki} + f^{-}_{ji}f^{+}_{ki}) 
\big(x_{jk}w^{+}_{jk}+(1-x_{jk})w^{-}_{jk}\big) \\&+
(f^{+}_{kj}f^{+}_{ij} + f^{+}_{kj}f^{-}_{ij} + f^{-}_{kj}f^{+}_{ij}) \big(x_{ki}w^{+}_{ki}+(1-x_{ki})w^{-}_{ki}\big)
\end{align*}
and
\begin{align}\label{eq:g_lp}
C_{LP} =& 
\sum_{i<j} 
\Big( 1- \sum_{k \in V-\{i,j\}} \Pr[A_{ij[k]}^{*}] \Big)
\big((1-x_{ij})w^{-}_{ij}+x_{ij}w^{+}_{ij}\big).
\end{align}
\begin{theorem}\label{thm:QuickCluster}
	\textsc{QuickCluster} is an $\alpha$-approximation algorithm for the weighted correlation clustering problem if the following two conditions hold:
	\begin{enumerate}
		\item For any pair of distinct vertices $i$ and $j$, we have
		$$ 
		\Delta(x_{ij}, w^{-}_{ij}, w^{+}_{ij}) \leq 0 
		$$
		where $\Delta(x_{ij}, w^{-}_{ij}, w^{+}_{ij}) = \big(f^{+}_{ij} w^{-}_{ij} + f^{-}_{ij} w^{+}_{ij}\big) - 
		\alpha \big((1-x_{ij})w^{-}_{ij}+x_{ij}w^{+}_{ij}\big)$.
		\item For any triple of distinct vertices $i$, $j$, and $k$, we have
		$$ 
		\Omega(x_{ij}, x_{jk}, x_{ki}, w^{-}_{ij}, w^{-}_{jk}, w^{-}_{ki}, w^{+}_{ij}, w^{+}_{jk}, w^{+}_{ki}) \leq 0 
		$$
		where $\Omega(x_{ij}, x_{jk}, x_{ki}, w^{-}_{ij}, w^{-}_{jk}, w^{-}_{ki}, w^{+}_{ij}, w^{+}_{jk}, w^{+}_{ki}) = 
		\phi_{ijk} - \alpha \cdot \psi_{ijk}$.
	\end{enumerate}
\end{theorem}
\begin{proof}
	If condition 1 holds, from equations \ref{eq:g_qcr} and \ref{eq:g_lp}, we have 
    $E\big[C_{QC} \big] \leq \alpha \cdot C_{LP}$.
	If condition 2 holds, from equations \ref{eq:b_qcr} and \ref{eq:b_lp}, we have 
	$E\big[UC_{QC} \big] \leq \alpha \cdot UC_{LP}$. Therefore
	$$
	E\big[ cost(QC) \big ] = 
	E\big[ C_{QC} \big] + E\big[ UC_{QC} \big] \leq 
	\alpha \cdot C_{LP} + \alpha \cdot UC_{LP} =
	\alpha \cdot cost(LP) \leq
	\alpha \cdot cost(OPT) 
	$$\end{proof}

To prove an approximation factor $\alpha$ for \textsc{QuickCluster} using Theorem \ref{thm:QuickCluster}, we need to show that the maximum values of $\Delta$ and $\Omega$ are less than or equal to 0 for that value of $\alpha$.
To achieve this, for the function $\Delta$, we find a subset of its domain that contains at least one of the points that maximizes the function. Then, we show that the maximum of this function on this subset is less than or equal to 0 which implies that the function on its whole domain is also less than or equal to 0. Similarly, we do the same for the function $\Omega$. Restricting the domain helps us finding the maximum value more easily.

If the instance satisfies the probability constraints then in both functions $\Delta$ and $\Omega$ we can omit the $w^+$ arguments and only consider the $w^{-}$ and $x$ arguments. Therefore, in the rest of this analysis we consider these functions as $\Delta(x_{ij}, w^{-}_{ij})$ and $\Omega(x_{ij}, x_{jk}, x_{ki}, w^{-}_{ij}, w^{-}_{jk}, w^{-}_{ki})$.

In $\Delta$, function $f^{-}_{ij}$ is either equal to $w^{-}_{ij}$ or $h(w^{-}_{ij})$. Consequently, for a fixed value of $w^{-}_{ij}$ the function $\Delta$ will be a linear function with a single variable $x_{ij}$, and receives its maximum value when $x_{ij}$ is in $\{0, 1\}$ which are the endpoints of the domain of $x_{ij}$. Hence, instead of the whole domain of $\Delta$, we can find the maximum value of $\Delta$ on 
\begin{align*}
	D(\Delta) = 
	\big\{
	( x_{ij} , w^{-}_{ij}) 
	\, \big| \, 
	x_{ij} \in \{0, 1\} \,\wedge\,
	w^{-}_{ij} \in [0, 1] 
	\big\}.
\end{align*}

Similarly, by fixing the values of $w^{-}_{ij}$, $w^{-}_{jk}$ and $w^{-}_{ki}$ in $\Omega$, it will be a linear function with three variables $x_{ij}$, $x_{jk}$ and $x_{ki}$.
According to the first constraint of LP \ref{LP:CWCC}, 
$(x_{ij},x_{jk},x_{ki}) \in S$ where 
$$S=\big\{ (a,b,c)\,\big|\, a,b,c \in [0, 1] \,\wedge\, a\leq b+c  \,\wedge\, b\leq c+a \,\wedge\, c\leq a+b \big\}.$$
Therefore, $\Omega$ receives its maximum value when $(x_{ij}, x_{jk}, x_{ki})$ is a vertex of the polytope $S$.
The set of these vertices is
$$L = \{ (0,0,0), (1,1,0), (1,0,1), (0,1,1), (1,1,1) \}.$$
Hence, instead of the whole domain of $\Omega$, we can find the maximum value of $\Omega$ on 
\begin{align*}
	D(\Omega) = 
	\big\{
	(x_{ij}, x_{jk}, x_{ki}, w^{-}_{ij}, w^{-}_{jk}, w^{-}_{ki}) \, \big| \,
	(x_{ij}, x_{jk}, x_{ki})\in L \,\wedge\, 
	w^{-}_{ij},w^{-}_{jk},w^{-}_{ki} \in [0,1]
	\big\}.
\end{align*}

\begin{theorem}\label{thm:QuickCluster_probability_constraint}
	\textit{QuickCluster} with $f^{-}_{ij} = w^{-}_{ij}$ is a 3-approximation algorithm for the weighted correlation clustering problem when the input instance satisfies the probability constraints.
\end{theorem}
\begin{proof}
We show that both conditions of Theorem \ref{thm:QuickCluster} hold for $\alpha=3$.
As discussed above, for the first condition it is enough to prove that for $\alpha=3$, $\Delta\leq0$ on domain $D(\Delta)$. In this domain, $x_{ij} \in \{0, 1\}$. For $x_{ij}=0$ we have 
\begin{align*}
\Delta(0, w^{-}_{ij}) &= 
f^{+}_{ij} w^{-}_{ij} + f^{-}_{ij} w^{+}_{ij}- 
3 w^{-}_{ij} \\ &=
w^{+}_{ij} w^{-}_{ij} + w^{-}_{ij} w^{+}_{ij} - 
3 w^{-}_{ij} \\ &=
2w^{+}_{ij} w^{-}_{ij} - 3 w^{-}_{ij} \\ &\leq
2w^{-}_{ij} - 3 w^{-}_{ij}
&&\text{(since $0 \leq w^{+}_{ij} \leq 1$.)} \\ &\leq 0.
\end{align*}
Similarly, it can be proved that when $x_{ij}=1$, we have $\Delta(1, w^{-}_{ij}) \leq 0$ which means that the first condition is valid for $\alpha=3$. 

For the second condition, note that according to the definition of $f$ and $w$ we have $f^{+}_{ij}=w^{+}_{ij}$, $f^{-}_{ij}=w^{-}_{ij}$, $w^{+}_{ij}=w^{+}_{ji}$ and $w^{-}_{ij}=w^{-}_{ji}$. Then, by simple algebraic operations, we have
\begin{align}\label{ineq:phi_p}
	\phi_{ijk} = 
	3(
	w^{-}_{ij}w^{+}_{jk}w^{+}_{ki}+
	w^{+}_{ij}w^{-}_{jk}w^{+}_{ki}+
	w^{+}_{ij}w^{+}_{jk}w^{-}_{ki}
	)
\end{align}
and
\begin{align*}
\psi_{ijk} =& 
    (w^{+}_{ik}w^{+}_{jk} + w^{+}_{ik}w^{-}_{jk} + w^{-}_{ik}w^{+}_{jk}) 
    \big(x_{ij}w^{+}_{ij}+(1-x_{ij})w^{-}_{ij}\big) \\&+
    (w^{+}_{ji}w^{+}_{ki} + w^{+}_{ji}w^{-}_{ki} + w^{-}_{ji}w^{+}_{ki}) 
    \big(x_{jk}w^{+}_{jk}+(1-x_{jk})w^{-}_{jk}\big) \\&+
    (w^{+}_{kj}w^{+}_{ij} + w^{+}_{kj}w^{-}_{ij} + w^{-}_{kj}w^{+}_{ij}) \big(x_{ki}w^{+}_{ki}+(1-x_{ki})w^{-}_{ki}\big).
\end{align*}
When $(x_{ij}, x_{jk}, x_{ki})\in L$, we have $\psi_{ijk} \geq 
w^{-}_{ij}w^{+}_{jk}w^{+}_{ki}+
w^{+}_{ij}w^{-}_{jk}w^{+}_{ki}+
w^{+}_{ij}w^{+}_{jk}w^{-}_{ki}$ which implies that on domain $D(\Omega)$, $\phi_{ijk} - 3 \cdot \psi_{ijk} \leq 0 $. Therefore, the second condition is also valid for $\alpha=3$.
\end{proof}

\begin{theorem}\label{thm:probability_and_triangle_constraint}
	\textit{QuickCluster} with $f^{-}_{ij} = h(w^{-}_{ij})$ is a 1.6-approximation algorithm for the weighted correlation clustering problem when the input instance satisfies both probability and triangle inequality constraints.
\end{theorem}
\begin{proof}
We show that both conditions of Theorem \ref{thm:QuickCluster} hold for $\alpha=1.6$ when both probability and triangle inequality constraints are satisfied. Then, according to Theorem \ref{thm:QuickCluster}, \textit{QuickCluster} is a 1.6-approximation algorithm for this problem.

For the first condition, we show that 
$\Delta\leq0$ on domain $D(\Delta)$ for $\alpha=1.6$:
\begin{itemize}
\item If $w^{-}_{ij} \in I_1$:
\begin{align*}
    \Delta(0, w^{-}_{ij})= - \frac{3}{5}w^{-}_{ij} \leq 0 \quad \textrm{and} \quad
    \Delta(1, w^{-}_{ij})= \frac{13}{5}w^{-}_{ij} - \frac{8}{5} \leq - \frac{69}{100}
\end{align*}
\item If $w^{-}_{ij} \in I_2$:
\begin{align*}
&\Delta(0, w^{-}_{ij})= - \frac{50}{7}(w^{-}_{ij})^{2} + \frac{383}{70}w^{-}_{ij} - \frac{5}{4} 
\leq - \frac{28311}{140000} \\
&\Delta(1, w^{-}_{ij})= - \frac{50}{7}(w^{-}_{ij})^{2} + \frac{607}{70}w^{-}_{ij} - \frac{57}{20}
\leq - \frac{30551}{140000}
\end{align*}
\item If $w^{-}_{ij} \in I_3$:
\begin{align*}
    \Delta(0, w^{-}_{ij})= 1 - \frac{13}{5} w^{-}_{ij} \leq - \frac{319}{500} \quad \textrm{and} \quad
    \Delta(1, w^{-}_{ij})= \frac{3}{5} w^{-}_{ij} - \frac{3}{5} \leq 0
\end{align*}
\end{itemize}

For the second condition, note that $(w^{-}_{ij}, w^{-}_{jk}, w^{-}_{ki}) \in S$ and the rounding function $h$ is a linear piece-wise function on intervals $I_1$, $I_2$ and $I_3$. Then, the partial derivative functions  
$\frac{\partial \Omega}{\partial w^{-}_{ij}}$,
$\frac{\partial \Omega}{\partial w^{-}_{ij}}$ and 
$\frac{\partial \Omega}{\partial w^{-}_{ij}}$
are step piece-wise functions having a constant value in each piece.
Therefore, $\Omega$ obtains its maximum value when the triangle inequality on $(w^{-}_{ij}, w^{-}_{jk}, w^{-}_{ki})$ is tight, or all parameters 
$w^{-}_{ij}$, $w^{-}_{jk}$ and $w^{-}_{ki}$ are in 
$\{ 0, 0.35, 0.63, 1\}$, the set of endpoints of $I_1$, $I_2$ and $I_3$.

On the other hand, due to the symmetry of $\Omega$, for any $x_1, x_2, x_3, w_1, w_2, w_3 \in[0, 1]$ we have
\begin{enumerate}
\item 
	$\Omega(x_1,x_2,x_3, w_1, w_2, w_3)=
	\Omega(x_3,x_1,x_2, w_3, w_1, w_2)=  
	\Omega(x_2,x_3,x_1, w_2, w_3, w_1).$
\item
	$\Omega(x_1,x_2,x_3, w_1, w_2, w_3)=
	\Omega(x_2,x_1,x_3, w_2, w_1, w_3).$
\end{enumerate}
The first equalities mean that to check the maximum value of $\Omega$, we only need to consider those points at which $w^{-}_{ki}$ is greater than or equal to both $w^{-}_{ij}$ and $w^{-}_{jk}$, and 
the second equality limits us to only consider those points at which
$x_{ij} \leq x_{jk}$.
So, when the instance satisfies both the probability and the triangle inequality constraints, it is enough to find the maximum value of $\Omega$ on 
\begin{align*}
D_T(\Omega) = 
\big\{
	(x_{ij}, x_{jk}, x_{ki}&, w^{-}_{ij}, w^{-}_{jk}, w^{-}_{ki})
	\, \big| \,
	(x_{ij}, x_{jk}, x_{ki}) \in L_T 
	\, \wedge \,
	(w^{-}_{ij}, w^{-}_{jk}, w^{-}_{ki}) \in A\,\vee\,B
\big\}
\end{align*}
where
\begin{align*}
&L_T=\big\{ (0,0,0), (1,1,0), (0,1,1), (1,1,1) \big\} \\
&A= \big\{ (w^{-}_{ij}, w^{-}_{jk}, w^{-}_{ki}) \in T \, \big| \, 
w^{-}_{ij}, w^{-}_{jk}, w^{-}_{ki} \in \{ 0, 0.35, 0.63, 1\} \big\}\\
&B=\big\{ (w^{-}_{ij}, w^{-}_{jk}, w^{-}_{ki}) \in T \, \big| \, 
w^{-}_{ij} + w^{-}_{jk} = w^{-}_{ki} \big\}
\end{align*}

Tab.~\ref{table:Omega} shows the values of $\Omega$ when 
$(x_{ij}, x_{jk}, x_{ki}) \in L_T$ 
and
$(w^{-}_{ij}, w^{-}_{jk}, w^{-}_{ki}) \in A$. 
The maximum value of $\Omega$ when 
$(x_{ij}, x_{jk}, x_{ki}) \in L_T$ 
and
$(w^{-}_{ij}, w^{-}_{jk}, w^{-}_{ki}) \in B$ has been computed in Appendix~\ref{AppendixA}. 
Since $\Omega$ is a piece-wise function, the maximum in each part must be calculated separately, which leads to a lengthy proof in Appendix~\ref{AppendixA}. As an informal but short proof, in Fig.~\ref{fig:omega} we showed the shape of $\Omega$ in that domain. 
The content of the table and detailed computations in this appendix show that $\Omega$ is less than or equal to zero in all its domain. Hence, the second condition of Theorem~\ref{thm:QuickCluster} for $\alpha = 1.6$ is also true.
\begin{figure}[h]
	\centering
	\includegraphics[width=0.9\textwidth]{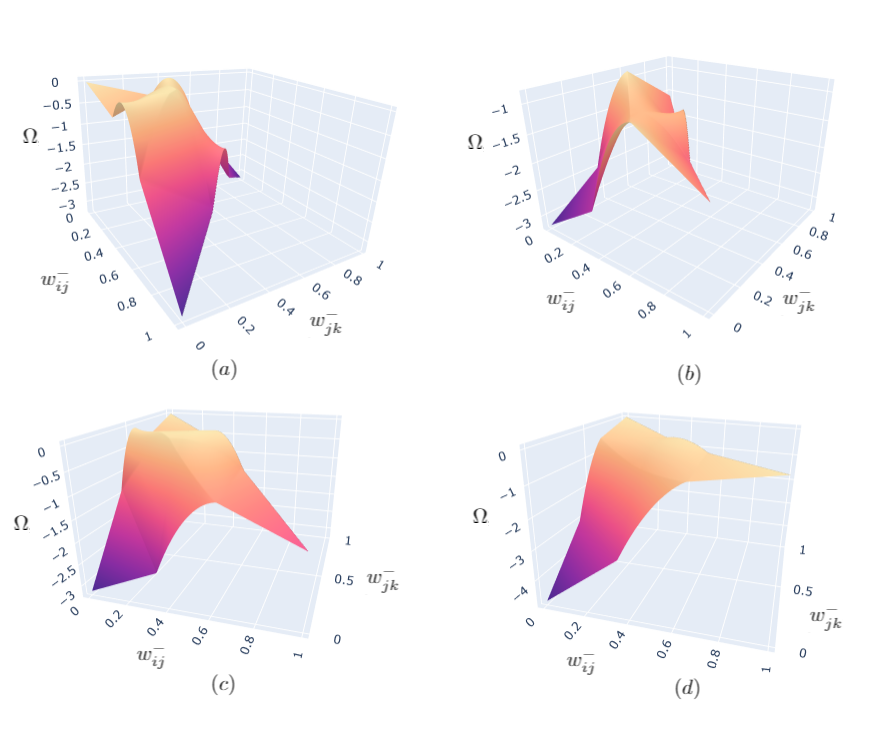}
	\caption{
	The shape of $\Omega$ for 
	$(a) $ $(x_{ij}, x_{jk}, x_{ki}) = (0,0,0)$,
	$(b) $ $(x_{ij}, x_{jk}, x_{ki}) = (1,1,0)$,
	$(c) $ $(x_{ij}, x_{jk}, x_{ki}) = (0,1,1)$, and 
	$(d) $ $(x_{ij}, x_{jk}, x_{ki}) = (1,1,1)$
	when $w^{-}_{ij} + w^{-}_{jk} = w^{-}_{ki}$.
	}
	\label{fig:omega}
\end{figure}
\end{proof}

\section{A Lower Bound for $\alpha$}

In this section, we want to discuss how efficient is the performance of the rounding functions we used for our algorithm. 
In particular, we want to show that if the rounding function $f_{ab}$ can be defined by $g(w^{-}_{ab})$ where $g: [0,1] \rightarrow [0,1]$ for any $a,b \in V$, then the second condition of Theorem \ref{alg:QuickCluster} can not hold for $\alpha<3$ if instances satisfy probability constraints, and $\alpha\leq1.55$ if both probability and triangle inequality be satisfied.

\begin{lemma}\label{lemma:limit1}
If the rounding function $f_{ab}$ be defined by $g(w^{-}_{ab})$ where 
$g: [0,1] \rightarrow [0,1]$, 
then $g(0) = 0$ and $g(1) = 1$, 
otherwise the second condition of Theorem \ref{thm:QuickCluster} does not hold for any value of $\alpha$.
\end{lemma}
\begin{proof}
If $f_{ab}=g(w^{-}_{ab})$ and 
$(x_{ij}, x_{jk}, x_{ki}, w^{-}_{ij}, w^{-}_{jk}, w^{-}_{ki}) = (0,0,0,0,0,0)$ 
then 
$\phi_{ijk} = 6 g(0) \left(1 - g(0)\right)$ and 
$\psi_{ijk} = 0$. If the second condition of Theorem \ref{thm:QuickCluster} holds then 
$ \Omega(0, 0, 0, 0, 0, 0) = 6 g(0) \left(1 - g(0)\right) \leq 0 $, and it means that $g(0)=0$.
Similarly  
$\Omega(1, 1, 1, 1, 1, 1) = 3 \left(1-g(1)\right)^{2} \leq 0$
if and only if $g(1)=1$.
\end{proof}

\begin{theorem}\label{thm:limit1}
For any rounding function $f_{ab}$ that can be defined by $g(w^{-}_{ab})$ where 
$g: [0,1] \rightarrow [0,1]$, 
the second condition of Theorem \ref{thm:QuickCluster} can not hold for any 
$\alpha<3$.
\end{theorem}
\begin{proof}
If $f_{ab}=g(w^{-}_{ab})$ and 
$(x_{ij}, x_{jk}, x_{ki}, w^{-}_{ij}, w^{-}_{jk}, w^{-}_{ki}) = (0, 0, 0, 1, 0, 0)$ 
then by Lemma \ref{lemma:limit1},
$\phi_{ijk} = 3$ and $\psi_{ijk} = 1$.
If the second condition of Theorem \ref{thm:QuickCluster} holds then 
$\Omega(0, 0, 0, 1, 0, 0) = 3 - \alpha \leq 0$. Hence $\alpha\geq 3$.
\end{proof}

Theorem \ref{thm:limit1} shows that the choice of 
$f_{ab}=w^{-}_{ab}$
is optimal for the instances that satisfy probability constraints. But the points we used in the proof of Theorem \ref{thm:limit1} do not satisfy triangle inequality. In the following we will find a lower bound for the instances that satisfy both probability and triangle inequality constraints.

\begin{lemma}\label{lemma:limit2}
 For any value of $\alpha$ and $f_{ij}=g(w^{-}_{ij})$ that satisfies the second condition of Theorem \ref{thm:QuickCluster}, and any value of $\beta$ such that $\alpha \leq \beta$, these following four cases hold.
\begin{enumerate}
\item If 
$g(x) \leq -x+1 $ 
and 
$ - \beta x + x^{2} - x + 1 \geq 0 $ then
$g(x) \leq - x - \sqrt{- \beta x + x^{2} - x + 1} + 1$
\item If $g(x) \leq -x+1 $
and 
$ \beta x - \beta + x^{2} - x + 1 \geq 0 $ then
$g(x) \leq - x - \sqrt{\beta x - \beta + x^{2} - x + 1} + 1$
\item If $ g(x) \geq -x+1 $
and 
$ - \beta x + x^{2} - x + 1 \geq 0 $ 
then
$g(x) \geq - x + \sqrt{- \beta x + x^{2} - x + 1} + 1$
\item If $ g(x) \geq -x+1 $
and 
$ \beta x - \beta + x^{2} - x + 1 \geq 0 $ then
$g(x) \geq - x + \sqrt{\beta x - \beta + x^{2} - x + 1} + 1$
\end{enumerate}
\end{lemma}
\begin{proof}
By Lemma \ref{lemma:limit1}, 
$
\Omega(0, 0, 0, x, x, 0) = - 2 \alpha x - 4 x g(x) + 2 x - 2 g(x)^{2} + 4 g(x)
$.
So, if the second condition of Theorem \ref{thm:QuickCluster} holds, then
\begin{align*}
    - 2 \alpha x - 4 x g(x) + 2 x - 2 g(x)^{2} + 4 g(x) &\leq 0 \\
    - \alpha x - 2 x g(x) + x - g(x)^{2} + 2 g(x) &\leq 0 \\
    - \alpha x &\leq g(x)^2 + 2 x g(x) -2 g(x) -x \\
    - \alpha x + [x^{2} - x + 1] &\leq g(x)^2 + 2 x g(x) -2 g(x) -x + [x^{2} -x + 1] \\
    - \alpha x + x^{2} - x + 1 &\leq (- g(x) - x + 1)^2
\end{align*}
Since $\beta \geq \alpha$, if $ - \beta x + x^{2} - x + 1 \geq 0 $ then $ - \alpha x + x^{2} - x + 1 \geq 0 $. So
\begin{itemize}
\item 
If 
$(- g(x) - x + 1)\geq0$ 
then 
$g(x) \leq - x - \sqrt{- \alpha x + x^{2} - x + 1} + 1$. For any fixed $x\in[0, 1]$, 
$L_1(\alpha)=- x - \sqrt{- \alpha x + x^{2} - x + 1} + 1$ is a non-decreasing function. So, we have the first case.
\item
If $(- g(x) - x + 1)\leq0$, then 
$g(x) \geq - x + \sqrt{- \alpha x + x^{2} - x + 1} + 1$. For any fixed $x\in[0, 1]$, 
$L_3(\alpha)=- x + \sqrt{- \alpha x + x^{2} - x + 1} + 1$ is a non-increasing function. So, we have the third case.
\end{itemize}

By Lemma \ref{lemma:limit1}, 
$
\Omega(1, 1, 0, x, x, 0) = 2 \alpha x - 2 \alpha - 2 g(x)^{2} - 4 x g(x) + 4 g(x) + 2 x
$.
So, when the second condition of Theorem \ref{thm:QuickCluster} holds we have
\begin{align*}
    2 \alpha x - 2 \alpha - 2 g(x)^{2} - 4 x g(x) + 4 g(x) + 2 x &\leq 0 \\
    \alpha x - \alpha - g(x)^{2} - 2 x g(x) + 2 g(x) + x &\leq 0 \\
    \alpha x - \alpha + x &\leq g(x)^{2} + 2 x g(x) -2 g(x)\\
    \alpha x - \alpha + x + [x^2 -2x +1] &\leq g(x)^{2} + 2 x g(x) -2 g(x) + [x^2 -2x +1]\\
    \alpha x - \alpha + x^{2} - x + 1 &\leq g(x)^2 + 2xg(x) -2g(x) + x^2 -2x +1\\
    \alpha x - \alpha + x^{2} - x + 1 &\leq (g(x) + x - 1)^2
\end{align*}
Since $\beta \geq \alpha$ and $x \in [0,1]$, if 
$ \beta x - \beta + x^{2} - x + 1 \geq 0 $ then 
$ \alpha x - \alpha + x^{2} - x + 1 \geq 0 $. So
\begin{itemize}
\item If  $(g(x) + x - 1) \leq 0$, then 
$g(x) \leq - x - \sqrt{\alpha x - \alpha + x^{2} - x + 1} + 1$. For any fixed $x\in[0, 1]$, 
$L_2(\alpha)=- x - \sqrt{\alpha x - \alpha + x^{2} - x + 1} + 1$ is a non-decreasing function. So, we have the third case.
\item If $(g(x) + x - 1) \geq 0$, then
$g(x) \geq - x + \sqrt{\alpha x - \alpha + x^{2} - x + 1} + 1$. For any fixed $x\in[0, 1]$, 
$L_4(\alpha)=- x + \sqrt{\alpha x - \alpha + x^{2} - x + 1} + 1$ is a non-increasing function. So, we have the fourth case.
\end{itemize}\end{proof}

Fig.~\ref{fig:lemma_5} (a) shows the regions which are restricted by Lemma~\ref{lemma:limit2} if $\alpha=1.55$. It means that any function $g$ that satisfies the second condition of Theorem~\ref{thm:QuickCluster} should not pass through the red regions. 
Fig.~\ref{fig:lemma_5} (b) shows the restricted regions for $\alpha=1.6$ and the shape of the function $h$. It can be seen that this function does not pass through the restricted regions.
\begin{figure}[h]
	\centering
	\includegraphics[width=\textwidth]{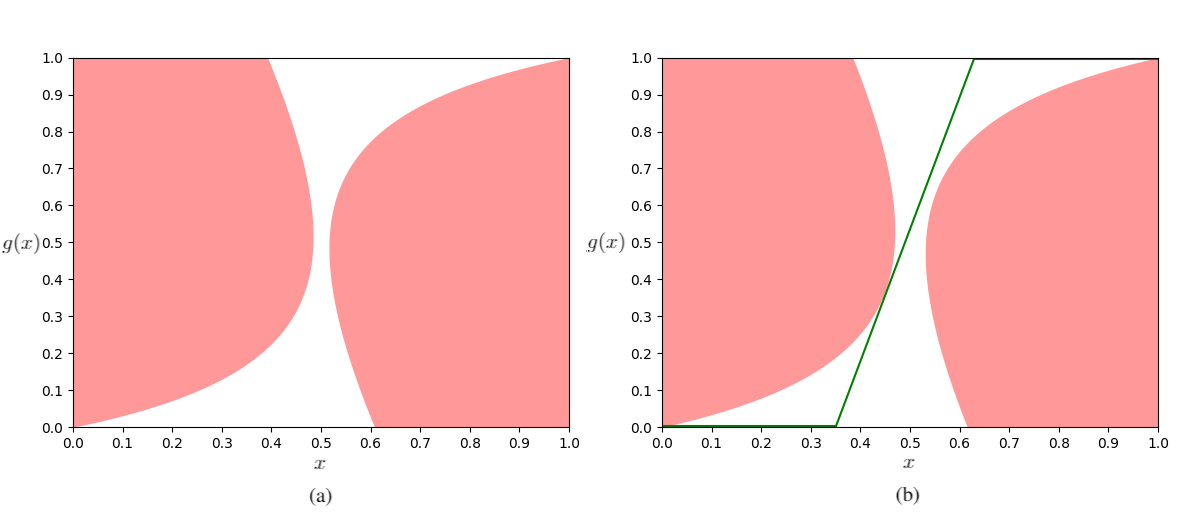}
	\caption{
	The red regions are those restricted by Lemma \ref{lemma:limit2} for 
	(a) $\alpha=1.55$ and 
	(b) $\alpha=1.6$. The green lines are the shape of $h$.
	}
	\label{fig:lemma_5}
\end{figure}

\begin{lemma}\label{lemma:limit3}
If $f_{ij}=g(w^{-}_{ij})$ satisfies the second condition of Theorem \ref{thm:QuickCluster}, then $g(0.44) < 0.293$ for any $\alpha\leq1.55$.
\end{lemma}
\begin{proof}
We define 
$\gamma_1(z)$ as $\Omega(0, 0, 0, 0.21, 0.23, 0.44)$
when $\alpha=z$. According to the definition of $\Omega$, we can conclude that $\gamma_1(z)$ is a decreasing function.
So, for fixed values of $g(0.21)$, $g(0.23)$, $g(0.44)$, and $\alpha\leq1.55$, we have
\begin{align}\label{eq:l2}
\gamma_1(\alpha) \geq& \gamma_1(1.55).
\end{align}
We can consider $g(0.21)$, $g(0.23)$, and $g(0.44)$ as the variables of $\gamma_1(1.55)$. 
If the second condition of Theorem \ref{thm:QuickCluster} holds for an $\alpha\leq1.55$ then 
$\gamma_1(\alpha)\leq0$. So by inequality \ref{eq:l2}, we have
\begin{align*}
\gamma_1(1.55) =& 
0.002 g(0.21) g(0.23) 
- 0.9535 g(0.21) g(0.44) 
+ 0.66 g(0.21) 
- 1.0445 g(0.23) g(0.44) \\&
+ 0.7 g(0.23) 
+ 1.12 g(0.44) 
- 0.4840 \leq 0.
\end{align*}
As a result
\begin{align}\label{eq:l4}
g(0.44) \leq 
\frac{4 \left(g(0.21) g(0.23) + 330 g(0.21) + 350 g(0.23) - 242\right)}
{1907 g(0.21) + 2089 g(0.23) - 2240}.
\end{align}
Let $K_1$ be the right side of inequality \ref{eq:l4}. Then
$$
\frac{\partial K_1}{\partial g(0.21)} = 
\frac{4 \left(2089 g(0.23)^{2} + 19680 g(0.23) - 277706\right)}
{(1907 g(0.21) + 2089 g(0.23) - 2240)^2}
$$
and
$$
\frac{\partial K_1}{\partial g(0.23)} = 
\frac{4 \left(1907 g(0.21)^{2} - 24160 g(0.21) - 278462\right)}
{(1907 g(0.21) + 2089 g(0.23) - 2240)^2}.
$$

By Lemma \ref{lemma:limit2}, 
$g(0.21) \in [0, 0.79-\sqrt{0.5086}]$ and 
$g(0.23) \in [0, 0.77- 2\sqrt{0.1166}]$. 
So, it can be verified easily that 
$\frac{\partial K_1}{\partial g(0.21)}<0$ and $\frac{\partial K_1}{\partial g(0.23)}<0$. 
Thus, $K_1$ is maximized when 
$g(0.21) = 0$ and $g(0.23) = 0$. The value of $K_1$ is less than $0.44$ on these points. 
So, $g(0.44) < 0.44$, and by the first case of Lemma~\ref{lemma:limit2}, we can conclude that 
$g(0.44) \leq 0.56 - 2\sqrt{0.0179} < 0.293$
for any $\alpha\leq1.55$.
\end{proof}

\begin{lemma}\label{lemma:limit4}
If $f_{ij}=g(w^{-}_{ij})$ satisfies the second condition of Theorem \ref{thm:QuickCluster}, then $g(0.44) > 0.309$ for any $\alpha\leq1.55$.
\end{lemma}
\begin{proof}
We define 
$\gamma_2(z)$ as $\Omega(1, 0, 1, 0.44, 0.28, 0.72)$ 
when $\alpha=z$. Similar to the proof of Lemma~\ref{lemma:limit3}, we have
\begin{align}\label{eq:l3}
\gamma_2(\alpha) \geq& \gamma_2(1.55).
\end{align}
If the second condition of Theorem \ref{thm:QuickCluster} holds for an $\alpha\leq1.55$, then 
$\gamma_2(\alpha)\leq0$. So by inequality \ref{eq:l3}, we have
\begin{align*}
\gamma_2(1.55) =& 
0.594 g(0.28) g(0.44) 
+ 0.188 g(0.28) g(0.72) 
- 0.320 g(0.28) \\&
- 0.726 g(0.44) g(0.72) 
+ 0.560 g(0.72) 
- 0.296 \leq 0.
\end{align*}
As a result
\begin{align}\label{eq:l5}
g(0.44) \geq
\frac{2 \left(- 47 g(0.28) g(0.72) + 80 g(0.28) - 140 g(0.72) + 74\right)}
{33 \cdot \left(9 g(0.28) - 11 g(0.72)\right)}.
\end{align}
Let $K_2$ be the right side of the inequality~\ref{eq:l5}. Then
$$
\frac{\partial K_2}{\partial g(0.28)} = 
\frac{2 \cdot \left(517 g(0.72)^{2} + 380 g(0.72) - 666\right)}
{33 \left(9 g(0.28) - 11 g(0.72)\right)^2}
$$
and
$$
\frac{\partial K_2}{\partial g(0.72)} = 
\frac{2 \left(- 423 g(0.28)^{2} - 380 g(0.28) + 814\right)}
{33 \left(9 g(0.28) - 11 g(0.72)\right)^2}.
$$
By Lemma \ref{lemma:limit2}, 
$g(0.28) \in [0, 0.72 - 2 \sqrt{0.0911}]$ and 
$g(0.72) \in [0.28 + 2 \sqrt{0.0911}, 1]$. So, it can be verified easily that 
$\frac{\partial K_2}{\partial g(0.28)}>0$ and $\frac{\partial K_2}{\partial g(0.72)}>0$. 
Thus $K_2$ is minimized when 
$g(0.28) = 0$ and 
$g(0.72) = 0.28 + 2 \sqrt{0.0911}$. The value of $K_2$ is greater than $0.309$ on these points. 
Therefore, if $f_{ij}=g(w^{-}_{ij})$ satisfies the second condition of Theorem \ref{thm:QuickCluster}, then $g(0.44) > 0.309$ for any $\alpha\leq1.55$.
\end{proof}

The result of Lemma~\ref{lemma:limit3} and Lemma~\ref{lemma:limit4} is that no function $g$ can satisfy the second condition of Theorem~\ref{thm:QuickCluster} for any $\alpha \leq 1.55$. Our presented function can support $\alpha = 1.6$ it means it is almost optimal.
Note that for the presented function $h$, we limited ourselves to piece-wise linear functions, in order to retain the required calculations of the second condition of the main theorem simple.

\section{Conclusion}
In this paper, we proposed an improved combinatorial approximation algorithm for the well-known weighted correlation clustering problem.
The approximation factor of our algorithm is 3 when the input instance satisfies the probability constraints and 1.6 when both the probability and the triangle inequality constraints are satisfied which are smaller than previous results. Although the approximation factors of our algorithm are slightly greater than the best-known LP-based algorithms, its smaller running time makes it superior and practical for real applications.

\bibliographystyle{abbrvnat}
\bibliography{sample-dmtcs}

\begin{thebibliography}{12}
\providecommand{\natexlab}[1]{#1}
\providecommand{\url}[1]{\texttt{#1}}
\expandafter\ifx\csname urlstyle\endcsname\relax
  \providecommand{\doi}[1]{doi: #1}\else
  \providecommand{\doi}{doi: \begingroup \urlstyle{rm}\Url}\fi

\bibitem[Ailon et~al.(2008)Ailon, Charikar, and Newman]{Ailon08}
N.~Ailon, M.~Charikar, and A.~Newman.
\newblock Aggregating inconsistent information: Ranking and clustering.
\newblock \emph{Journal of the ACM}, 55\penalty0 (5), 2008.
\newblock ISSN 0004-5411.
\newblock URL \url{https://doi.org/10.1145/1411509.1411513}.

\bibitem[Bansal et~al.(2004)Bansal, Blum, and Chawla]{Bansal04}
N.~Bansal, A.~Blum, and S.~Chawla.
\newblock Correlation clustering.
\newblock \emph{Machin Learning}, 56\penalty0 (1-3):\penalty0 89--113, 2004.
\newblock URL \url{https://doi.org/10.1023/B:MACH.0000033116.57574.95}.

\bibitem[Barak et~al.(2011)Barak, Raghavendra, and Steurer]{BRS}
B.~Barak, P.~Raghavendra, and D.~Steurer.
\newblock Rounding semidefinite programming hierarchies via global correlation.
\newblock In \emph{Proceedings of the 2011 IEEE 52nd Annual Symposium on
  Foundations of Computer Science}, FOCS '11, page 472–481, USA, 2011. IEEE
  Computer Society.
\newblock ISBN 9780769545714.
\newblock URL \url{https://doi.org/10.1109/FOCS.2011.95}.

\bibitem[Charikar et~al.(2005)Charikar, Guruswami, and Wirth]{CHARIKAR05}
M.~Charikar, V.~Guruswami, and A.~Wirth.
\newblock Clustering with qualitative information.
\newblock \emph{Journal of Computer and System Sciences}, 71\penalty0
  (3):\penalty0 360--383, 2005.
\newblock ISSN 0022-0000.
\newblock URL
  \url{https://www.sciencedirect.com/science/article/pii/S0022000004001424}.

\bibitem[Chawla et~al.(2015)Chawla, Makarychev, Schramm, and
  Yaroslavtsev]{NearOptimalClustering}
S.~Chawla, K.~Makarychev, T.~Schramm, and G.~Yaroslavtsev.
\newblock Near optimal lp rounding algorithm for correlation clustering on
  complete and complete k-partite graphs.
\newblock In \emph{Proceedings of the Forty-Seventh Annual ACM Symposium on
  Theory of Computing (STOC)}, page 219–228. Association for Computing
  Machinery, 2015.
\newblock ISBN 9781450335362.
\newblock URL \url{https://doi.org/10.1145/2746539.2746604}.

\bibitem[Cohen-Addad et~al.(2022)Cohen-Addad, Lee, and Newman]{Cohen-Addad}
V.~Cohen-Addad, E.~Lee, and A.~Newman.
\newblock Correlation clustering with {S}herali-{A}dams.
\newblock In \emph{63rd IEEE Annual Symposium on Foundations of Computer
  Science (FOCS)}, pages 651--661. IEEE Computer Society, 2022.
\newblock \doi{10.1109/FOCS54457.2022.00068}.

\bibitem[Cohen-Addad et~al.(2023)Cohen-Addad, Lee, Li, and
  Newman]{Cohen-Addad2}
V.~Cohen-Addad, E.~Lee, S.~Li, and A.~Newman.
\newblock { Handling Correlated Rounding Error via Preclustering: A
  $1.73$-approximation for Correlation Clustering}.
\newblock In \emph{64th IEEE Annual Symposium on Foundations of Computer
  Science (FOCS)}, pages 1082--1104. IEEE Computer Society, 2023.

\bibitem[Demaine et~al.(2006)Demaine, Emanuel, Fiat, and Immorlica]{DEMAINE}
E.~D. Demaine, D.~Emanuel, A.~Fiat, and N.~Immorlica.
\newblock Correlation clustering in general weighted graphs.
\newblock \emph{Theoretical Computer Science}, 361\penalty0 (2):\penalty0
  172--187, 2006.
\newblock ISSN 0304-3975.

\bibitem[Gionis et~al.(2007)Gionis, Mannila, and
  Tsaparas]{Clustering_Aggregation}
A.~Gionis, H.~Mannila, and P.~Tsaparas.
\newblock Clustering aggregation.
\newblock \emph{ACM Transactions on Knowledge Discovery from Data}, 1\penalty0
  (1), 2007.
\newblock ISSN 1556-4681.

\bibitem[Guruswami and Sinop(2011)]{GS}
V.~Guruswami and A.~K. Sinop.
\newblock Lasserre hierarchy, higher eigenvalues, and approximation schemes for
  graph partitioning and quadratic integer programming with psd objectives.
\newblock In \emph{2011 IEEE 52nd Annual Symposium on Foundations of Computer
  Science (FOCS)}, pages 482--491, 2011.
\newblock \doi{10.1109/FOCS.2011.36}.

\bibitem[Shamir et~al.(2004)Shamir, Sharan, and Tsur]{Shamir}
R.~Shamir, R.~Sharan, and D.~Tsur.
\newblock Cluster graph modification problems.
\newblock \emph{Discrete Applied Mathematics}, 144\penalty0 (1–2):\penalty0
  173–182, Nov. 2004.
\newblock ISSN 0166-218X.

\bibitem[Veldt(2022)]{veldt}
N.~Veldt.
\newblock Correlation clustering via strong triadic closure labeling: Fast
  approximation algorithms and practical lower bounds.
\newblock volume 162 of \emph{Proceedings of Machine Learning Research}, pages
  22060--22083. {PMLR}, 2022.

\end{thebibliography}
\label{sec:biblio}

\pagebreak

\appendix

\section{Appendix: Complementary data and computations omitted in the proof of Theorem~\ref{thm:probability_and_triangle_constraint}}
\label{AppendixA}

Tabel \ref{table:Omega} shows the values of $\Omega$ when 
$(x_{ij}, x_{jk}, x_{ki}) \in L_T$
and 
$(w^{-}_{ij}, w^{-}_{jk}, w^{-}_{ki}) \in A$. 
\begin{table}[h]
\centering
\caption{
The values of $\Omega$ when 
$(x_{ij}, x_{jk}, x_{ki}) \in L_T$ 
and
$(w^{-}_{ij}, w^{-}_{jk}, w^{-}_{ki}) \in A$.
}
\begin{tabular}{|cc|cccc|}
\hline
\multicolumn{2}{|c|}{\multirow{2}{*}{$\Omega$}} & \multicolumn{4}{c|}{$(x_{ij}, x_{jk}, x_{ki})$}                                   \\ \cline{3-6} 
\multicolumn{2}{|c|}{}  & \multicolumn{1}{c|}{(0,0,0)} & \multicolumn{1}{c|}{(1,1,0)} & \multicolumn{1}{c|}{(0,1,1)} & (1,1,1) \\ \hline
\multicolumn{1}{|c|}{}  & $( 0 , 0 , 0 )$ & \multicolumn{1}{c|}{ 0 }  & \multicolumn{1}{c|}{ -16/5 }  & \multicolumn{1}{c|}{ -16/5 }  &  -24/5        \\ \cline{2-6} 
\multicolumn{1}{|c|}{}  & $( 0 , 0.35 , 0.35 )$ & \multicolumn{1}{c|}{ -21/50 }  & \multicolumn{1}{c|}{ -5/2 }  & \multicolumn{1}{c|}{ -69/50 }  &  -149/50        \\ \cline{2-6} 
\multicolumn{1}{|c|}{}  & $( 0 , 0.63 , 0.63 )$ & \multicolumn{1}{c|}{ -319/250 }  & \multicolumn{1}{c|}{ -43/50 }  & \multicolumn{1}{c|}{ -111/250 }  &  -111/250        \\ \cline{2-6} 
\multicolumn{1}{|c|}{}  & $( 0 , 1 , 1 )$ & \multicolumn{1}{c|}{ -16/5 }  & \multicolumn{1}{c|}{ -8/5 }  & \multicolumn{1}{c|}{ 0 }  &  0        \\ \cline{2-6} 
\multicolumn{1}{|c|}{}  & $( 0.35 , 0 , 0.35 )$ & \multicolumn{1}{c|}{ -21/50 }  & \multicolumn{1}{c|}{ -5/2 }  & \multicolumn{1}{c|}{ -5/2 }  &  -149/50        \\ \cline{2-6} 
\multicolumn{1}{|c|}{}  & $( 0.35 , 0.35 , 0 )$ & \multicolumn{1}{c|}{ -21/50 }  & \multicolumn{1}{c|}{ -69/50 }  & \multicolumn{1}{c|}{ -5/2 }  &  -149/50        \\ \cline{2-6} 
\multicolumn{1}{|c|}{}  & $( 0.35 , 0.35 , 0.35 )$ & \multicolumn{1}{c|}{ -63/100 }  & \multicolumn{1}{c|}{ -159/100 }  & \multicolumn{1}{c|}{ -159/100 }  &  -207/100        \\ \cline{2-6} 
\multicolumn{1}{|c|}{}  & $( 0.35 , 0.35 , 0.63 )$ & \multicolumn{1}{c|}{ -99/500 }  & \multicolumn{1}{c|}{ -579/500 }  & \multicolumn{1}{c|}{ -131/500 }  &  -371/500        \\ \cline{2-6} 
\multicolumn{1}{|c|}{}  & $( 0.35 , 0.63 , 0.35 )$ & \multicolumn{1}{c|}{ -99/500 }  & \multicolumn{1}{c|}{ -131/500 }  & \multicolumn{1}{c|}{ -131/500 }  &  -371/500        \\ \cline{2-6} 
\multicolumn{1}{|c|}{}  & $( 0.35 , 0.63 , 0.63 )$ & \multicolumn{1}{c|}{ -319/250 }  & \multicolumn{1}{c|}{ -43/50 }  & \multicolumn{1}{c|}{ -111/250 }  &  -111/250        \\ \cline{2-6} 
\multicolumn{1}{|c|}{}  & $( 0.35 , 1 , 1 )$ & \multicolumn{1}{c|}{ -16/5 }  & \multicolumn{1}{c|}{ -8/5 }  & \multicolumn{1}{c|}{ 0 }  &  0        \\ \cline{2-6} 
\multicolumn{1}{|c|}{\multirow{3}{*}{$(w^{-}_{ij}, w^{-}_{jk}, w^{-}_{ki})$}}  
                        & $( 0.63 , 0 , 0.63 )$ & \multicolumn{1}{c|}{ -319/250 }  & \multicolumn{1}{c|}{ -43/50 }  & \multicolumn{1}{c|}{ -43/50 }  &  -111/250        \\ \cline{2-6} 
\multicolumn{1}{|c|}{}  & $( 0.63 , 0.35 , 0.35 )$ & \multicolumn{1}{c|}{ -99/500 }  & \multicolumn{1}{c|}{ -131/500 }  & \multicolumn{1}{c|}{ -579/500 }  &  -371/500        \\ \cline{2-6} 
\multicolumn{1}{|c|}{}  & $( 0.63 , 0.35 , 0.63 )$ & \multicolumn{1}{c|}{ -319/250 }  & \multicolumn{1}{c|}{ -43/50 }  & \multicolumn{1}{c|}{ -43/50 }  &  -111/250        \\ \cline{2-6} 
\multicolumn{1}{|c|}{}  & $( 0.63 , 0.63 , 0 )$ & \multicolumn{1}{c|}{ -319/250 }  & \multicolumn{1}{c|}{ -111/250 }  & \multicolumn{1}{c|}{ -43/50 }  &  -111/250        \\ \cline{2-6} 
\multicolumn{1}{|c|}{}  & $( 0.63 , 0.63 , 0.35 )$ & \multicolumn{1}{c|}{ -319/250 }  & \multicolumn{1}{c|}{ -111/250 }  & \multicolumn{1}{c|}{ -43/50 }  &  -111/250        \\ \cline{2-6} 
\multicolumn{1}{|c|}{}  & $( 0.63 , 0.63 , 0.63 )$ & \multicolumn{1}{c|}{ 0 }  & \multicolumn{1}{c|}{ 0 }  & \multicolumn{1}{c|}{ 0 }  &  0        \\ \cline{2-6} 
\multicolumn{1}{|c|}{}  & $( 0.63 , 0.63 , 1 )$ & \multicolumn{1}{c|}{ 0 }  & \multicolumn{1}{c|}{ 0 }  & \multicolumn{1}{c|}{ 0 }  &  0        \\ \cline{2-6} 
\multicolumn{1}{|c|}{}  & $( 0.63 , 1 , 0.63 )$ & \multicolumn{1}{c|}{ 0 }  & \multicolumn{1}{c|}{ 0 }  & \multicolumn{1}{c|}{ 0 }  &  0        \\ \cline{2-6} 
\multicolumn{1}{|c|}{}  & $( 0.63 , 1 , 1 )$ & \multicolumn{1}{c|}{ 0 }  & \multicolumn{1}{c|}{ 0 }  & \multicolumn{1}{c|}{ 0 }  &  0        \\ \cline{2-6} 
\multicolumn{1}{|c|}{}  & $( 1 , 0 , 1 )$ & \multicolumn{1}{c|}{ -16/5 }  & \multicolumn{1}{c|}{ -8/5 }  & \multicolumn{1}{c|}{ -8/5 }  &  0        \\ \cline{2-6} 
\multicolumn{1}{|c|}{}  & $( 1 , 0.35 , 1 )$ & \multicolumn{1}{c|}{ -16/5 }  & \multicolumn{1}{c|}{ -8/5 }  & \multicolumn{1}{c|}{ -8/5 }  &  0        \\ \cline{2-6} 
\multicolumn{1}{|c|}{}  & $( 1 , 0.63 , 0.63 )$ & \multicolumn{1}{c|}{ 0 }  & \multicolumn{1}{c|}{ 0 }  & \multicolumn{1}{c|}{ 0 }  &  0        \\ \cline{2-6} 
\multicolumn{1}{|c|}{}  & $( 1 , 0.63 , 1 )$ & \multicolumn{1}{c|}{ 0 }  & \multicolumn{1}{c|}{ 0 }  & \multicolumn{1}{c|}{ 0 }  &  0        \\ \cline{2-6} 
\multicolumn{1}{|c|}{}  & $( 1 , 1 , 0 )$ & \multicolumn{1}{c|}{ -16/5 }  & \multicolumn{1}{c|}{ 0 }  & \multicolumn{1}{c|}{ -8/5 }  &  0        \\ \cline{2-6} 
\multicolumn{1}{|c|}{}  & $( 1 , 1 , 0.35 )$ & \multicolumn{1}{c|}{ -16/5 }  & \multicolumn{1}{c|}{ 0 }  & \multicolumn{1}{c|}{ -8/5 }  &  0        \\ \cline{2-6} 
\multicolumn{1}{|c|}{}  & $( 1 , 1 , 0.63 )$ & \multicolumn{1}{c|}{ 0 }  & \multicolumn{1}{c|}{ 0 }  & \multicolumn{1}{c|}{ 0 }  &  0        \\ \cline{2-6} 
\multicolumn{1}{|c|}{}  & $( 1 , 1 , 1 )$ & \multicolumn{1}{c|}{ 0 }  & \multicolumn{1}{c|}{ 0 }  & \multicolumn{1}{c|}{ 0 }  &  0        \\ \hline 
\end{tabular}
\label{table:Omega}
\end{table}
\pagebreak

The maximum value of $\Omega$ when 
$(w^{-}_{ij}, w^{-}_{jk}, w^{-}_{ki}) \in B$
and
$(x_{ij}, x_{jk}, x_{ki}) \in L_T$ has been computed case by case as follows:
\\
$\bullet$ If $(w^{-}_{ij}, w^{-}_{jk}, w^{-}_{ki}) \in I_{ 1 } \times I_{ 1 } \times I_{ 1 } $: 
\begin{align*}
    &\Omega( 0 , 0 , 0 , w^{-}_{ij}, w^{-}_{jk}, w^{-}_{ij}+w^{-}_{jk}) =  - \frac{6 w^{-}_{ij}}{5} - \frac{6 w^{-}_{jk}}{5}  &
    &\text{argmax} = ( 0 , 0 , 0 ) &
    &\text{max} =  0
\\
    &\Omega( 1 , 1 , 0 , w^{-}_{ij}, w^{-}_{jk}, w^{-}_{ij}+w^{-}_{jk}) =  2 w^{-}_{ij} + 2 w^{-}_{jk} - \frac{16}{5}  &
    &\text{argmax} = ( 0 , \frac{7}{20} , \frac{7}{20} ) &
    &\text{max} =  - \frac{5}{2}
\\
    &\Omega( 0 , 1 , 1 , w^{-}_{ij}, w^{-}_{jk}, w^{-}_{ij}+w^{-}_{jk}) =  2 w^{-}_{ij} + \frac{26 w^{-}_{jk}}{5} - \frac{16}{5}  &
    &\text{argmax} = ( 0 , \frac{7}{20} , \frac{7}{20} ) &
    &\text{max} =  - \frac{69}{50}
\\
    &\Omega( 1 , 1 , 1 , w^{-}_{ij}, w^{-}_{jk}, w^{-}_{ij}+w^{-}_{jk}) =  \frac{26 w^{-}_{ij}}{5} + \frac{26 w^{-}_{jk}}{5} - \frac{24}{5}  &
    &\text{argmax} = ( 0 , \frac{7}{20} , \frac{7}{20} ) &
    &\text{max} =  - \frac{149}{50}
\end{align*}
\\
$\bullet$ If $(w^{-}_{ij}, w^{-}_{jk}, w^{-}_{ki}) \in I_{ 1 } \times I_{ 1 } \times I_{ 2 } $: 
\begin{align*}
    &\Omega( 0 , 0 , 0 , w^{-}_{ij}, w^{-}_{jk}, w^{-}_{ij}+w^{-}_{jk}) =  - \frac{50 \left(w^{-}_{ij}\right)^{2}}{7} - \frac{100 w^{-}_{ij} w^{-}_{jk}}{7} + \frac{591 w^{-}_{ij}}{70} - \frac{50 \left(w^{-}_{jk}\right)^{2}}{7} \\
    &\qquad\qquad\qquad\qquad\qquad\qquad\qquad
    + \frac{591 w^{-}_{jk}}{70} - \frac{5}{2}  \\
    &\text{argmax} = ( \frac{241}{1000} , \frac{7}{20} , \frac{591}{1000} ) \text{ \quad max} =  - \frac{719}{140000}
\\
    &\Omega( 1 , 1 , 0 , w^{-}_{ij}, w^{-}_{jk}, w^{-}_{ij}+w^{-}_{jk}) =  - \frac{50 \left(w^{-}_{ij}\right)^{2}}{7} - \frac{100 w^{-}_{ij} w^{-}_{jk}}{7} + \frac{163 w^{-}_{ij}}{14} - \frac{50 \left(w^{-}_{jk}\right)^{2}}{7} \\
    &\qquad\qquad\qquad\qquad\qquad\qquad\qquad
    + \frac{163 w^{-}_{jk}}{14} - \frac{57}{10}  \\
    &\text{argmax} = ( \frac{7}{25} , \frac{7}{20} , \frac{63}{100} ) \text{ \quad max} =  - \frac{6}{5}
\\
    &\Omega( 0 , 1 , 1 , w^{-}_{ij}, w^{-}_{jk}, w^{-}_{ij}+w^{-}_{jk}) =  - \frac{50 \left(w^{-}_{ij}\right)^{2}}{7} - \frac{100 w^{-}_{ij} w^{-}_{jk}}{7} + \frac{163 w^{-}_{ij}}{14} - \frac{50 \left(w^{-}_{jk}\right)^{2}}{7} \\
    &\qquad\qquad\qquad\qquad\qquad\qquad\qquad
    + \frac{1039 w^{-}_{jk}}{70} - \frac{57}{10}  \\
    &\text{argmax} = ( \frac{7}{25} , \frac{7}{20} , \frac{63}{100} ) \text{ \quad max} =  - \frac{2}{25}
\\
    &\Omega( 1 , 1 , 1 , w^{-}_{ij}, w^{-}_{jk}, w^{-}_{ij}+w^{-}_{jk}) =  - \frac{50 \left(w^{-}_{ij}\right)^{2}}{7} - \frac{100 w^{-}_{ij} w^{-}_{jk}}{7} + \frac{1039 w^{-}_{ij}}{70} - \frac{50 \left(w^{-}_{jk}\right)^{2}}{7} \\
    &\qquad\qquad\qquad\qquad\qquad\qquad\qquad
    + \frac{1039 w^{-}_{jk}}{70} - \frac{73}{10}  \\
    &\text{argmax} = ( \frac{7}{25} , \frac{7}{20} , \frac{63}{100} ) \text{ \quad max} =  - \frac{98}{125}
\end{align*}
\\
$\bullet$ If $(w^{-}_{ij}, w^{-}_{jk}, w^{-}_{ki}) \in I_{ 1 } \times I_{ 1 } \times I_{ 3 } $: 
\begin{align*}
    &\Omega( 0 , 0 , 0 , w^{-}_{ij}, w^{-}_{jk}, w^{-}_{ij}+w^{-}_{jk}) =  - \frac{16 w^{-}_{ij}}{5} - \frac{16 w^{-}_{jk}}{5} + 2  &
    &\text{argmax} = ( \frac{7}{25} , \frac{7}{20} , \frac{63}{100} ) &
    &\text{max} =  - \frac{2}{125}
\\
    &\Omega( 1 , 1 , 0 , w^{-}_{ij}, w^{-}_{jk}, w^{-}_{ij}+w^{-}_{jk}) =  - \frac{6}{5}  &
    &\text{argmax} = ( \frac{7}{25} , \frac{7}{20} , \frac{63}{100} ) &
    &\text{max} =  - \frac{6}{5}
\\
    &\Omega( 0 , 1 , 1 , w^{-}_{ij}, w^{-}_{jk}, w^{-}_{ij}+w^{-}_{jk}) =  \frac{16 w^{-}_{jk}}{5} - \frac{6}{5}  &
    &\text{argmax} = ( \frac{7}{25} , \frac{7}{20} , \frac{63}{100} ) &
    &\text{max} =  - \frac{2}{25}
\\
    &\Omega( 1 , 1 , 1 , w^{-}_{ij}, w^{-}_{jk}, w^{-}_{ij}+w^{-}_{jk}) =  \frac{16 w^{-}_{ij}}{5} + \frac{16 w^{-}_{jk}}{5} - \frac{14}{5}  &
    &\text{argmax} = ( \frac{7}{20} , \frac{7}{20} , \frac{7}{10} ) &
    &\text{max} =  - \frac{14}{25}
\end{align*}
\\
$\bullet$ If $(w^{-}_{ij}, w^{-}_{jk}, w^{-}_{ki}) \in I_{ 1 } \times I_{ 2 } \times I_{ 2 } $: 
\begin{align*}
    &\Omega( 0 , 0 , 0 , w^{-}_{ij}, w^{-}_{jk}, w^{-}_{ij}+w^{-}_{jk}) =  
    \frac{2875 \left(w^{-}_{ij}\right)^{2} w^{-}_{jk}}{49} - \frac{775 \left(w^{-}_{ij}\right)^{2}}{28} + \frac{2875 w^{-}_{ij} \left(w^{-}_{jk}\right)^{2}}{49} \\
    &\qquad\qquad\qquad\qquad\qquad\qquad\quad - \frac{9325 w^{-}_{ij} w^{-}_{jk}}{98} + \frac{16553 w^{-}_{ij}}{560} - \frac{1950 \left(w^{-}_{jk}\right)^{2}}{49} + \frac{1258 w^{-}_{jk}}{35} - \frac{65}{8} \\
    &\text{argmax} = ( \frac{241}{1000} , \frac{7}{20} , \frac{591}{1000} ) \text{ \quad max} =  - \frac{719}{140000}
\\
    &\Omega( 1 , 1 , 0 , w^{-}_{ij}, w^{-}_{jk}, w^{-}_{ij}+w^{-}_{jk}) =  \frac{125 \left(w^{-}_{ij}\right)^{2} w^{-}_{jk}}{7} - \frac{375 \left(w^{-}_{ij}\right)^{2}}{28} + \frac{125 w^{-}_{ij} \left(w^{-}_{jk}\right)^{2}}{7} \\
    &\qquad\qquad\qquad\qquad\qquad\qquad\quad - \frac{4525 w^{-}_{ij} w^{-}_{jk}}{98} + \frac{2309 w^{-}_{ij}}{112} - \frac{950 \left(w^{-}_{jk}\right)^{2}}{49} + \frac{174 w^{-}_{jk}}{7} - \frac{353}{40}  \\
    &\text{argmax} = ( 0 , \frac{63}{100} , \frac{63}{100} ) \text{ \quad max} =  - \frac{43}{50}
\\
    &\Omega( 0 , 1 , 1 , w^{-}_{ij}, w^{-}_{jk}, w^{-}_{ij}+w^{-}_{jk}) =  \frac{2875 \left(w^{-}_{ij}\right)^{2} w^{-}_{jk}}{49} - \frac{775 \left(w^{-}_{ij}\right)^{2}}{28} + \frac{2875 w^{-}_{ij} \left(w^{-}_{jk}\right)^{2}}{49} \\
    &\qquad\qquad\qquad\qquad\qquad\qquad\quad - \frac{9325 w^{-}_{ij} w^{-}_{jk}}{98} + \frac{3669 w^{-}_{ij}}{112} - \frac{1950 \left(w^{-}_{jk}\right)^{2}}{49} + \frac{1482 w^{-}_{jk}}{35} - \frac{453}{40}  \\
    &\text{argmax} = ( \frac{21}{100} , \frac{21}{50} , \frac{63}{100} ) \text{ \quad max} =  - \frac{69}{2000}
\end{align*}
\begin{align*}
    &\Omega( 1 , 1 , 1 , w^{-}_{ij}, w^{-}_{jk}, w^{-}_{ij}+w^{-}_{jk}) =  \frac{125 \left(w^{-}_{ij}\right)^{2} w^{-}_{jk}}{7} - \frac{375 \left(w^{-}_{ij}\right)^{2}}{28} + \frac{125 w^{-}_{ij} \left(w^{-}_{jk}\right)^{2}}{7} \\
    &\qquad\qquad\qquad\qquad\qquad\qquad\quad - \frac{4525 w^{-}_{ij} w^{-}_{jk}}{98} + \frac{13337 w^{-}_{ij}}{560} - \frac{950 \left(w^{-}_{jk}\right)^{2}}{49} + \frac{982 w^{-}_{jk}}{35} - \frac{417}{40}  \\
    &\text{argmax} = ( 0 , \frac{63}{100} , \frac{63}{100} ) \text{ \quad max} =  - \frac{111}{250}
\end{align*}
\\
$\bullet$ If $(w^{-}_{ij}, w^{-}_{jk}, w^{-}_{ki}) \in I_{ 1 } \times I_{ 2 } \times I_{ 3 } $: 
\begin{align*}
    &\Omega( 0 , 0 , 0 , w^{-}_{ij}, w^{-}_{jk}, w^{-}_{ij}+w^{-}_{jk}) =  \frac{15 w^{-}_{ij} w^{-}_{jk}}{7} - \frac{79 w^{-}_{ij}}{20} - \frac{50 \left(w^{-}_{jk}\right)^{2}}{7} - \frac{7 w^{-}_{jk}}{10} + 2  \\
    &\text{argmax} = ( \frac{7}{25} , \frac{7}{20} , \frac{63}{100} ) \text{ \quad max} =  - \frac{2}{125}
\\
    &\Omega( 1 , 1 , 0 , w^{-}_{ij}, w^{-}_{jk}, w^{-}_{ij}+w^{-}_{jk}) =  - \frac{65 w^{-}_{ij} w^{-}_{jk}}{7} + \frac{13 w^{-}_{ij}}{4} - \frac{50 \left(w^{-}_{jk}\right)^{2}}{7} + \frac{115 w^{-}_{jk}}{14} - \frac{16}{5}  \\
    &\text{argmax} = ( 0 , \frac{63}{100} , \frac{63}{100} ) \text{ \quad max} =  - \frac{43}{50}
\\
    &\Omega( 0 , 1 , 1 , w^{-}_{ij}, w^{-}_{jk}, w^{-}_{ij}+w^{-}_{jk}) =  \frac{15 w^{-}_{ij} w^{-}_{jk}}{7} - \frac{3 w^{-}_{ij}}{4} - \frac{50 \left(w^{-}_{jk}\right)^{2}}{7} + \frac{57 w^{-}_{jk}}{10} - \frac{6}{5}  \\
    &\text{argmax} = ( \frac{7}{20} , \frac{903}{2000} , \frac{1603}{2000} ) \text{ \quad max} =  - \frac{513}{80000}
\\
    &\Omega( 1 , 1 , 1 , w^{-}_{ij}, w^{-}_{jk}, w^{-}_{ij}+w^{-}_{jk}) =  - \frac{65 w^{-}_{ij} w^{-}_{jk}}{7} + \frac{129 w^{-}_{ij}}{20} - \frac{50 \left(w^{-}_{jk}\right)^{2}}{7} + \frac{799 w^{-}_{jk}}{70} - \frac{24}{5}  \\
    &\text{argmax} = ( \frac{7}{20} , \frac{1143}{2000} , \frac{1843}{2000} ) \text{ \quad max} =  - \frac{117351}{560000}
\end{align*}
\\
$\bullet$ If $(w^{-}_{ij}, w^{-}_{jk}, w^{-}_{ki}) \in I_{ 1 } \times I_{ 3 } \times I_{ 3 } $: 
\begin{align*}
    &\Omega( 0 , 0 , 0 , w^{-}_{ij}, w^{-}_{jk}, w^{-}_{ij}+w^{-}_{jk}) =  - \frac{13 w^{-}_{ij}}{5} - \frac{26 w^{-}_{jk}}{5} + 2  \\
    &\text{argmax} = ( 0 , \dfrac{63}{100} , \dfrac{63}{100} ) \text{\qquad max} =  \dfrac{-319}{250}
\\
    &\Omega( 1 , 1 , 0 , w^{-}_{ij}, w^{-}_{jk}, w^{-}_{ij}+w^{-}_{jk}) =  - \frac{13 w^{-}_{ij}}{5} - 2 w^{-}_{jk} + \frac{2}{5}  \\
    &\text{argmax} = ( 0 , \dfrac{63}{100} , \dfrac{63}{100} ) \text{\qquad max} =  \dfrac{-43}{50}
\end{align*}
\begin{align*}
    &\Omega( 0 , 1 , 1 , w^{-}_{ij}, w^{-}_{jk}, w^{-}_{ij}+w^{-}_{jk}) =  \frac{3 w^{-}_{ij}}{5} + \frac{6 w^{-}_{jk}}{5} - \frac{6}{5}  \\
    &\text{argmax} = ( 0 , 1 , 1 ) \text{ \qquad max} =  0
\\
    &\Omega( 1 , 1 , 1 , w^{-}_{ij}, w^{-}_{jk}, w^{-}_{ij}+w^{-}_{jk}) =  \frac{3 w^{-}_{ij}}{5} + \frac{6 w^{-}_{jk}}{5} - \frac{6}{5}  \\
    &\text{argmax} = ( 0 , 1 , 1 ) \text{ \qquad max} =  0
\end{align*}
\\
$\bullet$ If $(w^{-}_{ij}, w^{-}_{jk}, w^{-}_{ki}) \in I_{ 2 } \times I_{ 1 } \times I_{ 2 } $: 
\begin{align*}
    &\Omega( 0 , 0 , 0 , w^{-}_{ij}, w^{-}_{jk}, w^{-}_{ij}+w^{-}_{jk}) =  \frac{2875 \left(w^{-}_{ij}\right)^{2} w^{-}_{jk}}{49} - \frac{1950 \left(w^{-}_{ij}\right)^{2}}{49} + \frac{2875 w^{-}_{ij} \left(w^{-}_{jk}\right)^{2}}{49} \\
    &\qquad\qquad\qquad\quad - \frac{9325 w^{-}_{ij} w^{-}_{jk}}{98} + \frac{1258 w^{-}_{ij}}{35} - \frac{775 \left(w^{-}_{jk}\right)^{2}}{28} + \frac{16553 w^{-}_{jk}}{560} - \frac{65}{8}  \\
    &\text{argmax} = ( \frac{7}{20} , \frac{241}{1000} , \frac{591}{1000} ) \text{ \quad max} =  - \frac{719}{140000}
\\
    &\Omega( 1 , 1 , 0 , w^{-}_{ij}, w^{-}_{jk}, w^{-}_{ij}+w^{-}_{jk}) =  \frac{125 \left(w^{-}_{ij}\right)^{2} w^{-}_{jk}}{7} - \frac{950 \left(w^{-}_{ij}\right)^{2}}{49} + \frac{125 w^{-}_{ij} \left(w^{-}_{jk}\right)^{2}}{7} \\
    &\qquad\qquad\qquad\quad - \frac{4525 w^{-}_{ij} w^{-}_{jk}}{98} + \frac{174 w^{-}_{ij}}{7} - \frac{375 \left(w^{-}_{jk}\right)^{2}}{28} + \frac{2309 w^{-}_{jk}}{112} - \frac{353}{40}  \\
    &\text{argmax} = ( \frac{63}{100} , 0 , \frac{63}{100} ) \text{ \quad max} =  - \frac{43}{50}
\\
    &\Omega( 0 , 1 , 1 , w^{-}_{ij}, w^{-}_{jk}, w^{-}_{ij}+w^{-}_{jk}) =  \frac{125 \left(w^{-}_{ij}\right)^{2} w^{-}_{jk}}{7} - \frac{950 \left(w^{-}_{ij}\right)^{2}}{49} + \frac{125 w^{-}_{ij} \left(w^{-}_{jk}\right)^{2}}{7} \\
    &\qquad\qquad\qquad\quad - \frac{4525 w^{-}_{ij} w^{-}_{jk}}{98} + \frac{174 w^{-}_{ij}}{7} - \frac{375 \left(w^{-}_{jk}\right)^{2}}{28} + \frac{13337 w^{-}_{jk}}{560} - \frac{353}{40}  \\
    &\text{argmax} = ( \frac{7}{20} , \frac{7}{25} , \frac{63}{100} ) \text{ \quad max} =  - \frac{38}{125}
\\
    &\Omega( 1 , 1 , 1 , w^{-}_{ij}, w^{-}_{jk}, w^{-}_{ij}+w^{-}_{jk}) =  \frac{125 \left(w^{-}_{ij}\right)^{2} w^{-}_{jk}}{7} - \frac{950 \left(w^{-}_{ij}\right)^{2}}{49} + \frac{125 w^{-}_{ij} \left(w^{-}_{jk}\right)^{2}}{7} \\
    &\qquad\qquad\qquad\quad - \frac{4525 w^{-}_{ij} w^{-}_{jk}}{98} + \qquad \frac{982 w^{-}_{ij}}{35} - \frac{375 \left(w^{-}_{jk}\right)^{2}}{28} + \frac{13337 w^{-}_{jk}}{560} - \frac{417}{40}  \\
    &\text{argmax} = ( \frac{63}{100} , 0 , \frac{63}{100} ) \text{ \quad max} =  - \frac{111}{250}
\end{align*}
\\
$\bullet$ If $(w^{-}_{ij}, w^{-}_{jk}, w^{-}_{ki}) \in I_{ 2 } \times I_{ 1 } \times I_{ 3 } $: 
\begin{align*}
    &\Omega( 0 , 0 , 0 , w^{-}_{ij}, w^{-}_{jk}, w^{-}_{ij}+w^{-}_{jk}) =  - \frac{50 \left(w^{-}_{ij}\right)^{2}}{7} + \frac{15 w^{-}_{ij} w^{-}_{jk}}{7} - \frac{7 w^{-}_{ij}}{10} - \frac{79 w^{-}_{jk}}{20} + 2  \\
    &\text{argmax} = ( \frac{7}{20} , \frac{7}{25} , \frac{63}{100} ) \text{ \quad max} =  - \frac{2}{125}
\\
    &\Omega( 1 , 1 , 0 , w^{-}_{ij}, w^{-}_{jk}, w^{-}_{ij}+w^{-}_{jk}) =  - \frac{50 \left(w^{-}_{ij}\right)^{2}}{7} - \frac{65 w^{-}_{ij} w^{-}_{jk}}{7} + \frac{115 w^{-}_{ij}}{14} + \frac{13 w^{-}_{jk}}{4} - \frac{16}{5}  \\
    &\text{argmax} = ( \frac{63}{100} , 0 , \frac{63}{100} ) \text{ \quad max} =  - \frac{43}{50}
\\
    &\Omega( 0 , 1 , 1 , w^{-}_{ij}, w^{-}_{jk}, w^{-}_{ij}+w^{-}_{jk}) =  - \frac{50 \left(w^{-}_{ij}\right)^{2}}{7} - \frac{65 w^{-}_{ij} w^{-}_{jk}}{7} + \frac{115 w^{-}_{ij}}{14} + \frac{129 w^{-}_{jk}}{20} - \frac{16}{5}  \\
    &\text{argmax} = ( \frac{7}{20} , \frac{7}{20} , \frac{7}{10} ) \text{ \quad max} =  - \frac{2}{25}
\\
    &\Omega( 1 , 1 , 1 , w^{-}_{ij}, w^{-}_{jk}, w^{-}_{ij}+w^{-}_{jk}) =  - \frac{50 \left(w^{-}_{ij}\right)^{2}}{7} - \frac{65 w^{-}_{ij} w^{-}_{jk}}{7} + \frac{799 w^{-}_{ij}}{70} + \frac{129 w^{-}_{jk}}{20} - \frac{24}{5}  \\
    &\text{argmax} = ( \frac{1143}{2000} , \frac{7}{20} , \frac{1843}{2000} ) \text{ \quad max} =  - \frac{117351}{560000}
\end{align*}
\\
$\bullet$ If $(w^{-}_{ij}, w^{-}_{jk}, w^{-}_{ki}) \in I_{ 2 } \times I_{ 2 } \times I_{ 3 } $: 
\begin{align*}
    &\Omega( 0 , 0 , 0 , w^{-}_{ij}, w^{-}_{jk}, w^{-}_{ij}+w^{-}_{jk}) =  \frac{2875 \left(w^{-}_{ij}\right)^{2} w^{-}_{jk}}{49} - \frac{775 \left(w^{-}_{ij}\right)^{2}}{28} + \frac{2875 w^{-}_{ij} \left(w^{-}_{jk}\right)^{2}}{49} \\
    &\qquad\qquad\qquad\qquad\quad
    - \frac{6105 w^{-}_{ij} w^{-}_{jk}}{98} + \frac{8213 w^{-}_{ij}}{560} - \frac{775 \left(w^{-}_{jk}\right)^{2}}{28} + \frac{8213 w^{-}_{jk}}{560} - \frac{9}{8}  \\
    &\text{argmax} = ( \frac{7}{20} , \frac{7}{20} , \frac{7}{10} ) \text{ \quad max} =  - \frac{6}{25}
\\
    &\Omega( 1 , 1 , 0 , w^{-}_{ij}, w^{-}_{jk}, w^{-}_{ij}+w^{-}_{jk}) =  \frac{2875 \left(w^{-}_{ij}\right)^{2} w^{-}_{jk}}{49} - \frac{775 \left(w^{-}_{ij}\right)^{2}}{28} + \frac{2875 w^{-}_{ij} \left(w^{-}_{jk}\right)^{2}}{49} \\
    &\qquad\qquad\qquad\qquad\quad
    - \frac{8345 w^{-}_{ij} w^{-}_{jk}}{98} + \frac{3089 w^{-}_{ij}}{112} - \frac{775 \left(w^{-}_{jk}\right)^{2}}{28} + \frac{3089 w^{-}_{jk}}{112} - \frac{333}{40}  \\
    &\text{argmax} = ( \frac{7}{20} , \frac{7}{20} , \frac{7}{10} ) \text{ \quad max} =  - \frac{6}{5}
\end{align*}
\begin{align*}
    &\Omega( 0 , 1 , 1 , w^{-}_{ij}, w^{-}_{jk}, w^{-}_{ij}+w^{-}_{jk}) =  \frac{125 \left(w^{-}_{ij}\right)^{2} w^{-}_{jk}}{7} - \frac{375 \left(w^{-}_{ij}\right)^{2}}{28} + \frac{125 w^{-}_{ij} \left(w^{-}_{jk}\right)^{2}}{7} \\
    &\qquad\qquad\qquad\qquad\quad 
    - \frac{2425 w^{-}_{ij} w^{-}_{jk}}{98} + \frac{183 w^{-}_{ij}}{16} - \frac{375 \left(w^{-}_{jk}\right)^{2}}{28} + \frac{7237 w^{-}_{jk}}{560} - \frac{153}{40}  \\
    &\text{argmax} = ( \frac{7}{20} , \frac{903}{2000} , \frac{1603}{2000} ) \text{ \quad max} =  - \frac{513}{80000}
\\
    &\Omega( 1 , 1 , 1 , w^{-}_{ij}, w^{-}_{jk}, w^{-}_{ij}+w^{-}_{jk}) =  \frac{125 \left(w^{-}_{ij}\right)^{2} w^{-}_{jk}}{7} - \frac{375 \left(w^{-}_{ij}\right)^{2}}{28} + \frac{125 w^{-}_{ij} \left(w^{-}_{jk}\right)^{2}}{7} \\
    &\qquad\qquad\qquad\qquad\quad 
    - \frac{3545 w^{-}_{ij} w^{-}_{jk}}{98} + \frac{1491 w^{-}_{ij}}{80} - \frac{375 \left(w^{-}_{jk}\right)^{2}}{28} + \frac{1491 w^{-}_{jk}}{80} - \frac{297}{40}  \\
    &\text{argmax} = ( \frac{1}{2} , \frac{1}{2} , 1 ) \text{ \quad max} =  - \frac{247}{3920}
\end{align*}
\\
$\bullet$ If $(w^{-}_{ij}, w^{-}_{jk}, w^{-}_{ki}) \in I_{ 2 } \times I_{ 3 } \times I_{ 3 } $: 
\begin{align*}
    &\Omega( 0 , 0 , 0 , w^{-}_{ij}, w^{-}_{jk}, w^{-}_{ij}+w^{-}_{jk}) =  \frac{\left(100 w^{-}_{ij} - 63\right) \left(13 w^{-}_{ij} + 26 w^{-}_{jk} - 10\right)}{140}  \\
    &\text{argmax} = ( \frac{37}{100} , \frac{63}{100} , 1 ) \text{ \quad max} =  - \frac{14547}{7000}
\\
    &\Omega( 1 , 1 , 0 , w^{-}_{ij}, w^{-}_{jk}, w^{-}_{ij}+w^{-}_{jk}) =  \frac{\left(100 w^{-}_{ij} - 63\right) \left(13 w^{-}_{ij} + 10 w^{-}_{jk} - 2\right)}{140}  \\
    &\text{argmax} = ( \frac{37}{100} , \frac{63}{100} , 1 ) \text{ \quad max} =  - \frac{11843}{7000}
\\
    &\Omega( 0 , 1 , 1 , w^{-}_{ij}, w^{-}_{jk}, w^{-}_{ij}+w^{-}_{jk}) =  \frac{3 \cdot \left(100 w^{-}_{ij} - 63\right) \left(- w^{-}_{ij} - 2 w^{-}_{jk} + 2\right)}{140}  \\
    &\text{argmax} = ( \frac{37}{100} , \frac{63}{100} , 1 ) \text{ \quad max} =  - \frac{1443}{7000}
\\
    &\Omega( 1 , 1 , 1 , w^{-}_{ij}, w^{-}_{jk}, w^{-}_{ij}+w^{-}_{jk}) =  \frac{3 \cdot \left(100 w^{-}_{ij} - 63\right) \left(- w^{-}_{ij} - 2 w^{-}_{jk} + 2\right)}{140}  \\
    &\text{argmax} = ( \frac{37}{100} , \frac{63}{100} , 1 ) \text{ \quad max} =  - \frac{1443}{7000}
\end{align*}
\\
$\bullet$ If $(w^{-}_{ij}, w^{-}_{jk}, w^{-}_{ki}) \in I_{ 3 } \times I_{ 1 } \times I_{ 3 } $: 
\begin{align*}
    &\Omega( 0 , 0 , 0 , w^{-}_{ij}, w^{-}_{jk}, w^{-}_{ij}+w^{-}_{jk}) =  - \frac{26 w^{-}_{ij}}{5} - \frac{13 w^{-}_{jk}}{5} + 2  \\
    &\text{argmax} = ( \frac{63}{100} , 0 , \frac{63}{100} ) \text{ \quad max} =  - \frac{319}{250}
\\
    &\Omega( 1 , 1 , 0 , w^{-}_{ij}, w^{-}_{jk}, w^{-}_{ij}+w^{-}_{jk}) =  -2 w^{-}_{ij} - \frac{13 w^{-}_{jk}}{5} + \frac{2}{5}  \\
    &\text{argmax} = ( \frac{63}{100} , 0 , \frac{63}{100} ) \text{ \quad max} =  - \frac{43}{50}
\\
    &\Omega( 0 , 1 , 1 , w^{-}_{ij}, w^{-}_{jk}, w^{-}_{ij}+w^{-}_{jk}) =  -2 w^{-}_{ij} + \frac{3 w^{-}_{jk}}{5} + \frac{2}{5} \\
    &\text{argmax} = ( \frac{63}{100} , \frac{7}{20} , \frac{49}{50} ) \text{ \quad max} =  - \frac{13}{20}
\\
    &\Omega( 1 , 1 , 1 , w^{-}_{ij}, w^{-}_{jk}, w^{-}_{ij}+w^{-}_{jk}) =  \frac{6 w^{-}_{ij}}{5} + \frac{3 w^{-}_{jk}}{5} - \frac{6}{5} \\
    &\text{argmax} = ( 1 , 0 , 1 ) \text{ \quad max} =  0
\end{align*}
\\
$\bullet$ If $(w^{-}_{ij}, w^{-}_{jk}, w^{-}_{ki}) \in I_{ 3 } \times I_{ 2 } \times I_{ 3 } $: 
\begin{align*}
    &\Omega( 0 , 0 , 0 , w^{-}_{ij}, w^{-}_{jk}, w^{-}_{ij}+w^{-}_{jk}) =  \frac{\left(100 w^{-}_{jk} - 63\right) \left(26 w^{-}_{ij} + 13 w^{-}_{jk} - 10\right)}{140}  \\
    &\text{argmax} = ( \frac{63}{100} , \frac{37}{100} , 1 ) \text{ \quad max} =  - \frac{14547}{7000} \\
    &\Omega( 1 , 1 , 0 , w^{-}_{ij}, w^{-}_{jk}, w^{-}_{ij}+w^{-}_{jk}) =  \frac{\left(100 w^{-}_{jk} - 63\right) \left(10 w^{-}_{ij} + 13 w^{-}_{jk} - 2\right)}{140}  \\
    &\text{argmax} = ( \frac{63}{100} , \frac{37}{100} , 1 ) \text{ \quad max} =  - \frac{11843}{7000} \\
    &\Omega( 0 , 1 , 1 , w^{-}_{ij}, w^{-}_{jk}, w^{-}_{ij}+w^{-}_{jk}) =  \frac{\left(100 w^{-}_{jk} - 63\right) \left(10 w^{-}_{ij} - 3 w^{-}_{jk} - 2\right)}{140}  \\
    &\text{argmax} = ( \frac{63}{100} , \frac{37}{100} , 1 ) \text{ \quad max} =  - \frac{4147}{7000} \\
    &\Omega( 1 , 1 , 1 , w^{-}_{ij}, w^{-}_{jk}, w^{-}_{ij}+w^{-}_{jk}) =  \frac{3 \cdot \left(100 w^{-}_{jk} - 63\right) \left(- 2 w^{-}_{ij} - w^{-}_{jk} + 2\right)}{140}  \\
    &\text{argmax} = ( \frac{63}{100} , \frac{37}{100} , 1 ) \text{ \quad max} =  - \frac{1443}{7000}
\end{align*}

\end{document}